\documentclass[11pt,letterpaper]{article}

\usepackage{latexsym,multirow}
\usepackage{amssymb,amsmath,bm}
\usepackage{mathtools}
\usepackage{graphicx}
\usepackage[color,all,import,arrow]{xy}
\usepackage{enumerate}
\usepackage{booktabs}
\usepackage{enumitem}

\usepackage{natbib}
\usepackage{xcolor}
\usepackage[pdftex, bookmarksopen=true, bookmarksnumbered=true,
pdfstartview=FitH, breaklinks=true, urlbordercolor={0 1 0}, citebordercolor={0 0 1}, colorlinks=true, citecolor=blue!50!black, linkcolor=blue!50!black, urlcolor=blue!50!black]{hyperref}
\usepackage{colortbl}
\usepackage{subfigure}
\usepackage{subcaption}

\usepackage{tikz}
\usepackage{inputenc}
\usetikzlibrary{shapes,arrows,trees,fit,positioning,shapes.misc,calc}

\usepackage{dcolumn}
\newcolumntype{.}{D{.}{.}{-1}}
\newcolumntype{d}[1]{D{.}{.}{#1}}

\usepackage{theorem}
\theoremstyle{plain}
\theoremheaderfont{\scshape}
\newtheorem{assumption}{Assumption}
\newtheorem{example}{Example}
\newtheorem{definition}{Definition}

\newtheorem{proposition}{Proposition}
\newtheorem{theorem}{Theorem}
\newtheorem{remark}{Remark}

\newtheorem{lemma}{Lemma}
\newtheorem{axiom}{Axiom}
\newenvironment{proof}{\paragraph{Proof:}}{\hfill$\square$}

\newcommand{\argmax}{\operatornamewithlimits{argmax}}

\newcommand{\qed}{\hfill \ensuremath{\Box}}

\usepackage{rotating}

\usepackage[compact]{titlesec}

\allowdisplaybreaks

\newcommand\spacingset[1]{\renewcommand{\baselinestretch}%
  {#1}\small\normalsize}

\newcommand{\blind}{0}

\newcommand{\cX}{\mathcal{X}}

\newcommand{\cA}{\mathcal{A}}
\newcommand{\cB}{\mathcal{B}}

\newcommand{\bv}{\bm{v}}
\newcommand{\bx}{\bm{x}}
\newcommand{\bX}{\bm{X}}
\newcommand{\cM}{\mathcal{M}}

\newcommand{\cZ}{\mathcal{Z}}

\newcommand{\cD}{\mathcal{D}}

\newcommand{\cY}{\mathcal{Y}}

\newcommand{\by}{\bm{y}}

\newcommand{\ext}{\mathrm{Ext}}

\newcommand{\E}{\mathbb{E}}
\newcommand{\R}{\mathbb{R}}
\newcommand{\bR}{\mathbb{R}}

\usepackage{dsfont}
\usepackage{caption}
\usepackage{tikz-cd}
\usetikzlibrary{arrows.meta}

\newcommand{\lt}{\left}
\newcommand{\rt}{\right}
\newcommand{\indicator}[1]{\mathds{1}{\left\{ #1 \right\}}}

\newcommand{\cS}{\mathcal{S}}

\newcommand{\indep}{\mbox{$\perp\!\!\!\perp$}}

\begin{document}

\newcommand{\tit}{An Axiomatic Foundation for Decisions with Counterfactual Utility}
%
%
\spacingset{1.25}

\if0\blind

{\title{\bf\tit\thanks{We thank an anonymous reviewer of Harvard IQSS's rapidPeer for posing the question this paper answers. We also thank Isaiah Andrews, Peter Buisseret, Philip Dawid, Andrew Gelman, Amanda Kowalski, Jonas Magdy Mikhaeil, Amit Sawant, Stephen Senn and Mats Julius Stensrud for helpful discussions.}}

\author{Benedikt Koch\thanks{Ph.D. candidate, Department of Statistics, Harvard University. 33 Oxford Street, Cambridge MA 02138, U.S.A. Email: \href{mailto:benedikt_koch@g.harvard.edu}{benedikt\_koch@g.harvard.edu} URL: \href{https://benediktjkoch.github.io}{https://benediktjkoch.github.io}} \and Kosuke
  Imai\thanks{Professor, Department of Government and Department of
    Statistics, Harvard University.  1737 Cambridge Street,
    Institute for Quantitative Social Science, Cambridge MA 02138, U.S.A.
    Email: \href{mailto:imai@harvard.edu}{imai@harvard.edu} URL:
    \href{https://imai.fas.harvard.edu}{https://imai.fas.harvard.edu}} \and Tomasz Strzalecki\thanks{Professor, Department of Economics, Harvard University, 1805 Cambridge Street, Cambridge MA 02138, U.S.A. Email: \href{mailto:tomasz_strzalecki@harvard.edu}{tomasz\_strzalecki@harvard.edu} URL:
    \href{https://tomasz.scholars.harvard.edu}{https://tomasz.scholars.harvard.edu}} }
}

\date{
\today
}

\maketitle

\fi

\if1\blind
\title{\bf \tit}

\maketitle
\fi

\pdfbookmark[1]{Title Page}{Title Page}

\thispagestyle{empty}
\setcounter{page}{0}

\begin{abstract}
  Counterfactual utilities evaluate decisions not only by the realized
  outcome under a given decision, but also by the counterfactual
  outcomes that would arise under alternative decisions. By
  generalizing the standard utility framework, they allow decision-makers
  to encode asymmetric criteria, such as avoiding harm and
  anticipating regret. Recent work, however, has raised fundamental
  concerns about the coherence and transitivity of counterfactual
  utilities. We address these concerns by extending the von Neumann--Morgenstern (vNM) framework to preferences defined on the extended space of all potential outcomes rather than realized outcomes alone. We show that expected counterfactual utility satisfies
  the vNM axioms on this extended domain, thereby admitting a coherent preference representation.
  We further examine how
  counterfactual preferences map onto the realized outcome space
  through menu-dependent and context-dependent projections.  This
  axiomatic framework reconciles apparent inconsistencies highlighted
  by the Russian roulette example in the statistics literature and
  resolves the well-known Allais paradox from behavioral economics. We
  also derive an additional axiom required to reduce counterfactual
  utilities to standard utilities on the same potential outcome space,
  and establish an axiomatic foundation for additive counterfactual
  utilities, which satisfy a necessary and sufficient condition for
  point identification. Our results show that the use of asymmetric 
  counterfactual utility functions is coherent regardless of whether 
  individual potential outcomes are deterministic or stochastic, and 
  that a recent proposal involving stochastic potential outcomes leads 
  to ambiguities and inconsistencies.
  
\end{abstract}


\newpage

\section{Introduction}\label{sec:introduction}

Suppose a physician must decide whether to administer a new,
experimental drug to a patient.  Let $D_i \in \{0,1\}$ denote a binary
treatment indicator, where $D_i = 1$ corresponds to the experimental
drug and $D_i = 0$ to the standard treatment.  We are interested in
patient survival following the treatment decision, denoted by
$Y_i \in \{0,1\}$, where $Y_i = 1$ indicates survival and $Y_i = 0$
indicates death.

An oracle policy would administer the experimental drug only to patients
who benefit from it. In the potential-outcomes framework
\citep{neyman1923application,rubin1974estimating}, this corresponds to
treating patient $i$ if and only if $Y_i(1)=1$ and $Y_i(0)=0$, where
$Y_i(d)$ denotes the outcome that would occur under treatment $d$.
In practice, however, such policies are not directly implementable because of the
fundamental problem of causal inference: we observe only one potential
outcome, $Y_i = Y_i(D_i)$ \citep{holland1986}.

Given this limitation, scholars typically base decisions solely on
realized outcomes and seek to maximize average, rather than
individual-level, causal effects.  Formally, a decision-maker specifies
a utility function $u(d; y_d)$ that depends on a given action $d$ and
its realized outcome $y_d$. We refer to such utilities as
\emph{standard utilities}.  Under this setting, an optimal decision
rule maximizes the expected utility
$V(d; u) := \mathbb{E}[u(d; Y(d))]$
\citep[e.g.,][]{manski2004statistical}.

Recently, a number of scholars have considered \emph{counterfactual
  utilities} that directly incorporate individual causal effects
\citep[e.g.,][]{li2019unit,kallus2022s,
  NEURIPS2022_ebcf1bff, mueller2023personalized,
 christy2024starting, christy2026countingdefiersdesignbasedmodel,
ben-michaelPolicyLearningAsymmetric,  koch2025statisticaldecisiontheorycounterfactual, imai2026triage}. In this
framework, the decision-maker specifies a utility function
$\tilde{u}(d; y_0,y_1)$ that depends on both potential outcomes.
Evaluation is then based on the expected counterfactual utility
$V(d; \tilde{u}) = \mathbb{E}[\tilde{u}(d; Y(0),Y(1))]$
\citep{li2019unit,koch2025statisticaldecisiontheorycounterfactual}. Counterfactual utility fits naturally within \cite{wald1950statistical}'s statistical decision framework, where the state of nature is the joint distribution of the full potential outcome vector.

This formulation allows utility to be defined at the individual level
as a function of both potential outcomes, thereby permitting a variety
of considerations to be explicitly incorporated.  For example,
asymmetric counterfactual utilities can reflect the Hippocratic
principle of ``do no harm’’ by assigning greater disutility to causing
a patient's death through the administration of the experimental drug
than to failing to prevent death by continuing with the standard
treatment
\citep[e.g.,][]{bordley2009hippocratic, NEURIPS2022_ebcf1bff,christy2024starting,ben-michaelPolicyLearningAsymmetric}.

Recently, however, some scholars have raised fundamental questions
about the coherence of counterfactual utilities \citep[e.g.,][]{dawid2023personalised,sarvet2023perspectives}. This paper is motivated by \citet{Gelman_2025} (hereafter, GM), which acknowledges the appeal
of asymmetric counterfactual utilities but raises the concern that
decisions based on them need not maximize survival. Using a Russian
roulette example, GM points out that an optimal rule under an
asymmetric counterfactual utility may select a treatment with a lower
survival probability over one with a higher survival probability.
Based on this observation, GM concludes that counterfactual utilities
can yield nonsensical recommendations, such as money pumps, because
they ``violate the axioms of classical ([von] Neumann--Morgenstern)
decision theory'' (page~2). In
addition, \citet{sawant2026counterfactualharm} argues
that decisions based on counterfactual utilities can violate
transitivity, which is one of the von Neumann--Morgenstein axioms.

In this paper, we establish an axiomatic foundation for counterfactual
statistical decision theory
(Section~\ref{sec:axiomatizing_counterfactual_decisiont_theory}). 
Specifically, we prove the coherence of decision making based on counterfactual utilities by generalizing the vNM \citep{NeumannMorgenstern1944} axioms to the extended space where preferences are defined on the {\it potential outcome space}. We further show that
counterfactual utilities can violate the vNM axioms if preferences
are instead defined over realized outcomes on the {\it realized outcome space}. These results explain the aforementioned controversy in the literature by distinguishing the space over which preferences are defined.
In addition, we directly connect counterfactual and standard decision theories by introducing an additional axiom under which counterfactual utilities collapse to standard utilities on this potential outcome space. 

We also show that counterfactual utility can be 
used descriptively and resolve the well-known Allais paradox in the
behavioral economics literature \citep{allais1953comportement}. This is achieved by generalizing regret theory \citep[e.g.,][]{bell1982regret, loomes1982regret, lanzani2022correlation}. Unlike the original regret theory, our generalization yields transitive preferences which are defined on the potential outcome space.

Furthermore, we derive an axiomatic characterization of {\it additive}
counterfactual utilities (Section~\ref{sec:identification}). A central
criticism of counterfactual utilities is that their expectations are
generally not identifiable \citep[e.g.,][]{dawid2023personalised,sarvet2023perspectives}, as they depend on the joint distribution
of potential outcomes. However, \cite{koch2025statisticaldecisiontheorycounterfactual} has shown that additive counterfactual utilities
depend only on the marginal distributions of potential outcomes and
constitute the largest class of counterfactual utilities whose
expectations are identifiable.  
Thus, this axiomatization establishes a direct link between preferences on the potential outcome space and the point identification of expected counterfactual utility.

We next show that our axiomatic results hold regardless of whether
potential outcomes are treated as deterministic or stochastic
(Section~\ref{sec:spo}). This analysis is motivated by GM’s proposal to extend
counterfactual utilities by taking expectations over stochastic
potential outcomes before evaluating decisions. We first show that this extension is not
unique, giving rise to ambiguity. Second, under mild conditions, every extended counterfactual utility has an equivalent extended standard utility that yields the same decision recommendation, thereby eliminating the asymmetry of counterfactual utilities. We further show that this
extended utility violates the vNM independence axiom on the potential
outcome space and the same vNM violation persists after projecting the
induced preferences onto the realized outcome space.  The only way to eliminate these inconsistencies is to
impose additivity, which yields additive counterfactual utilities.

Finally, to connect counterfactual utilities to the behavioral economics literature, we formalize two projections that map counterfactual utilities into preferences on the realized outcome space
(Section~\ref{sec:mapping_counterfactual_preferences_to_choice}). The
first is the {\it menu-dependent} projection, of which regret
theory is a special case.  \cite{sawant2026counterfactualharm} considered this projection.  We show that under this projection,
counterfactual utilities may violate transitivity. In contrast, the {\it context-dependent} projection always induces
a transitive preference order.

\section{Examples of Asymmetric Counterfactual Utilities} \label{sec:examples}

To motivate our theoretical development in later sections, we first revisit GM’s Russian roulette example. We show that, when placed in a criminal justice context, its conclusion no longer appears nonsensical. This suggests that, like standard utilities, the interpretation of counterfactual utilities depends on the decision setting.  Furthermore, we illustrate the descriptive value of counterfactual utilities using the well-known \cite{allais1953comportement} paradox from behavioral economics. We show that asymmetric counterfactual utilities can explain a common pattern of decision-making observed in practice rather than serving solely as a normative framework.

\subsection{Russian Roulette}\label{sec:russian_roulette}

GM considers the game of Russian roulette, in which one must choose between two lotteries:
\begin{itemize}
\item $R_{1/6}$: survive with probability $5/6$ and die with probability $1/6$,
\item $R_{1/7}$: survive with probability $6/7$ and die with probability $1/7$.
\end{itemize}
In this case, $R_{1/7}$ is clearly preferred, as it offers a higher probability of survival. Let $d=0$ and $d=1$ denote $R_{1/6}$ and $R_{1/7}$, respectively, and let $Y=1$ represent survival and $Y=0$ death. Under the standard utility function $u(d; y_d)=y_d$, expected utility therefore ranks $R_{1/7}$ above $R_{1/6}$, i.e., $V(1; u) > V(0;u)$.

GM introduces the following asymmetric counterfactual utility function, under which choosing a decision that ``saves'' the life under $R_{1/7}$ yields only half the gain of saving the life under $R_{1/6}$, i.e.,
\begin{align}\label{eqn:gm_utility}
    \tilde u(0; y_0, y_1) = \indicator{y_0 > y_1}, \quad \text{and} \quad
    \tilde u(1; y_0, y_1) = 0.5\, \indicator{y_0 < y_1}.
\end{align}
The decision is determined by the expected utility difference $V(1; \tilde u) - V(0; \tilde u)$.
To evaluate this quantity, GM assumes that the two potential outcomes are independent:
\begin{assumption}[Independent Potential Outcomes]\label{ass:independence}
    Y(0) \indep Y(1).
\end{assumption}
Under this assumption, GM shows that
$V(1; \tilde u) - V(0; \tilde u) = -\frac{1}{21} < 0$,
leading to the nonsensical conclusion that $R_{1/6}$ should be preferred to $R_{1/7}$.

We now show that this seemingly nonsensical result becomes reasonable in a different setting. Suppose that a judge at a pretrial hearing must decide whether to release an arrestee on their own recognizance ($d=0$) or impose cash bail ($d=1$), which may result in pretrial detention. The outcome of interest is whether the arrestee is rearrested for a new crime before a trial ($Y=0$) or not ($Y=1$) \citep{imaiExperimentalEvaluationAlgorithmassisted2023,ben-michaelDoesAIHelp2024,imai2026triage}. Under the same numerical specification as in the Russian roulette example, $6/7 \approx 86\%$ of arrestees would not be rearrested under cash bail, compared to $5/6 \approx 83\%$ under release.

Under Assumption~\ref{ass:independence}, cash bail would \emph{prevent} rearrest for $1/7\approx14\%$ of arrestees, while causing $5/42\approx12\%$ of arrestees to be rearrested who otherwise would not have been, imposing an unnecessary burden on them. Thus, although cash bail yields a lower average rearrest rate, a decision-maker may place greater weight on avoiding the harms of coercive state intervention than on realizing its preventive benefits \citep{dobbie2018effects}. This asymmetric tradeoff between individual liberty and public safety has been discussed in the pretrial literature \citep[e.g.][]{stevenson2022pretrial} and recognized by the Supreme Court (see \textit{United States v. Salerno}).  Such a consideration can be naturally encoded by a counterfactual utility of the form given in Equation~\eqref{eqn:gm_utility}.  \citet{mueller2023personalized,10.1093/aje/kwae441} provides a related medical example. 

This example underscores the fact that a utility function represents the preferences of the decision-maker and is therefore depends on the decision setting, in which it is applied. A utility specification that appears reasonable in one setting may seem nonsensical in another.
The same point also applies to standard utilities. For example, consider the standard utility function $u(d;y_d)=y_d-c \cdot d$ where $c$ represents a cost parameter.  This utility function recommends $R_{1/6}$ for a sufficiently large value of $c$.  Although such utilities may be reasonable when the decision is costly, they make little sense in the Russian roulette example.
Moreover, even within the same setting, different decision-makers may adopt distinct utility functions, reflecting heterogeneous ethical, clinical, or practical considerations.

\subsection{The Allais Paradox}

Next, through the well-known
\cite{allais1953comportement} paradox, we illustrate the descriptive value of asymmetric
counterfactual utility.
\cite{kahneman1979prospect} report the results of two experiments. In
the first experiment, subjects are asked to choose one of the
following two lotteries.
\begin{itemize}
\item[$(a_1)$.] Receive \$4000 with $0.8$ probability.
\item[$(b_1)$.] Receive \$3000 with certainty.
\end{itemize}
The authors show that most people choose option~($b_1$) over option~($a_1$).  In the second experiment, subjects are asked to choose between the following slightly modified options,
\begin{itemize}
\item[$(a_2)$.] Receive \$4000 with $0.2$ probability.
\item[$(b_2)$.] Receive \$3000 with $0.25$ probability.
\end{itemize}
The authors find that in this experiment, most people select option~($a_2$) instead of option~($b_2$).

This empirical finding is inconsistent with the standard utility.
To see this, assume that the utility depends only on the monetary value they receive, regardless of which option one chooses, i.e., $u(d; y_d) = u(y_d)$. Without loss of generality, we normalize $u(0) = 0$. Then, the result of the first experiment implies $V(b_1; u) > V(a_1; u) \Leftrightarrow u(3000) > 0.8 u(4000)$. In contrast, the result of the second experiment implies $V(a_2; u) > V(b_2; u) \Leftrightarrow u(3000) < 0.8 u(4000)$. This leads to a contradiction regardless of what values $u(4000)$ and $u(3000)$ take.

The behavioral economics literature has proposed many models to
explain the Allais paradox and other related phenomena (e.g., reflection effect, probabilistic insurance, and preference reversals). One prominent approach is regret theory \citep[e.g.,][]{bell1982regret, loomes1982regret, bikhchandani2011transitive, lanzani2022correlation}. Suppose that a subject experiences ``regret'' when forgoing the option that would have yielded a better outcome, and has ``rejoicing'' when the chosen option is better than the alternative. This idea can be formalized through the following asymmetric counterfactual utility \citep{bell1982regret}
\begin{align}\label{eqn:bell_utility}
    \tilde u_{\mathrm{Bell}}(d; y_0, y_1) = y_d + f_\lambda(y_d - y_{1-d})
    \quad \text{where} \quad
    f_\lambda(r) = 1 - \exp{(-\lambda r)}, \lambda > 0.
\end{align}
If $r > 0$, then $f_\lambda(r) > 0$, which captures rejoicing. If $r < 0$, then $f_\lambda(r) < 0$, representing regret. Because $|f_\lambda(-r)| > |f_\lambda(r)|$ for all $r > 0$, the utility encodes that a loss is evaluated more severely than a gain of the same magnitude. Under the assumption of independent lotteries (i.e., Assumption~\ref{ass:independence}), this asymmetric counterfactual utility can lead to the decisions consistent with the above experimental results for certain parameter values (e.g., $\lambda \geq 0.003$).
In Section~\ref{sec:mapping_counterfactual_preferences_to_choice}, we show that counterfactual decision theory encompasses regret theory as a special case.

\section{Axiomatizing Counterfactual Decision Theory}\label{sec:axiomatizing_counterfactual_decisiont_theory}

The preceding examples illustrate the normative and descriptive value of counterfactual utilities. In this section, we show that counterfactual utilities induce a coherent decision-theoretic framework, which satisfies the von Neumann--Morgenstern axioms \citep{NeumannMorgenstern1944} in the \emph{potential outcome space}. We then introduce an additional axiom under which the framework reduces to standard utilities.


\subsection{Setup}\label{sec:setup}

Consider a setting in which a decision $D \in \cD := \{0,1,\ldots,K-1\}$ is chosen for a unit with potential outcomes $(Y(0),\ldots,Y(K-1)) \in \cY^{\cD}$ and covariates $\bX \in \cX$, where $\cY := \{0,1,\ldots,M-1\}$ and $K,M \ge 2$. We assume $\cX$ is finite. Let $\cZ := \cD \times \cY^{\cD} \times \cX$ denote the {\it potential outcome space}.

Our aim is to evaluate the quality of a (possibly randomized) \emph{policy}
$\pi:\cY^{\cD}\times\cX \to \Delta(\cD)$, where $\Delta(\cD)$ denotes the set of probability distributions on $\cD$. Our framework asks the decision maker to evaluate all policies, including oracle policies that depends on the full vector of potential outcomes, though in practice a policy will only depend on covariates, i.e., $\pi(d; y_0, \ldots, y_{K-1}, \bx) = \pi(d; \bx)$. We use the extended evaluative domain because it allows us to fully express the decision maker's preferences, including counterfactual ones. Under this setting, each policy $\pi$ induces a decision defined by
$$D \mid  Y(0) = y_0, \ldots, Y(K-1) = y_{K-1}, \bX = \bx  \sim \pi(\cdot ; y_0, \ldots, y_{K-1}, \bx).$$

To quantify the consequence of choosing $d$ for a unit with potential outcomes $\by=(y_0,\ldots,y_{K-1})$ and covariates $\bx$, we specify a \emph{counterfactual utility} function $\tilde u: \cZ \to \R$ and evaluate a policy based on its expectation,
\begin{align}\label{eqn:value_fct}
    V_P(\pi; \tilde u) & := \ \E_{P^\pi}[\tilde u(D; Y(0), \ldots, Y(K-1), \bX)] \notag
    \\
    & \ = \ \sum_{d \in \cD} \sum_{\by \in \cY^\cD} \sum_{\bx \in \cX} \tilde u(d; \by, \bx) \cdot P^\pi(D = d, Y(0) = y_0, \ldots, Y(K-1) = y_{K-1}, \bX = \bx) \notag
    \\
    & \ = \ \sum_{d \in \cD} \sum_{\by \in \cY^\cD} \sum_{\bx \in \cX} \tilde u(d; \by, \bx) \cdot \pi(d; \by, \bx) \cdot P(Y(0) = y_0, \ldots, Y(K-1) = y_{K-1}, \bX = \bx),
\end{align}
where $P^\pi(d, \by, \bx) = \pi(d; \by, \bx) \cdot P(\by, \bx)$ is a joint law of $Z = (D, Y(0), \ldots, Y(K-1), \bX) \in \cZ$ induced by the policy $\pi$ and the probability measure $P$. The policy $\pi$ determines the conditional distribution of the decision given $(\by, \bx)$, while $P$ specifies the joint distribution of potential outcomes and covariates. We treat $P$ as the unknown \emph{state of nature} \citep[e.g.,][]{wald1950statistical, berger1985statistical, manski2004statistical, STOYE200970, stoye2011axioms, koch2025statisticaldecisiontheorycounterfactual}. In the terminology of classical decision theory, a policy $\pi$ corresponds to an \emph{act} that maps a state of the world $P$ to an outcome distribution $P^\pi$, i.e., $\pi(P) = P^\pi$ \citep[e.g.,][] {stoye2011axioms,stoye2011statistical}.

Let $\Delta(\cZ)$ denote the set of probability measures on $\cZ$. Because we allow for an arbitrary oracle policy, every element of $\Delta(\cZ)$ can be written as $P^\pi$ for some pair $(\pi,P)$. This is interpreted as the distribution induced by deploying policy $\pi$ in state $P$. Given a utility function $\tilde u$, we evaluate $P^\pi$ by its \emph{value} $V_P(\pi;\tilde u)$, which induces a preference relation on $\Delta(\cZ)$.
\begin{definition}[Preference Relation]\label{def:repr}
  Let $P^\pi,Q^\rho \in \Delta(\cZ)$ denote the probability distributions induced by policies $\pi$ and $\rho$ under states of nature $P$ and $Q$, respectively.
  Then, a binary preference relation $\succsim$ on $\Delta(\cZ)$ is represented by a counterfactual utility $\tilde u$ if the value of policy $\pi$ in state $P$ is no less than that of policy $\rho$ in state $Q$, i.e.,
  \[
    P^\pi \succsim Q^\rho
    \iff V_P(\pi;\tilde u) \ge V_Q(\rho;\tilde u).
  \]
\end{definition}
For a given utility $\tilde u$, Definition~\ref{def:repr} states that the decision maker (weakly) prefers deploying $\pi$ in state $P$ over deploying $\rho$ in state $Q$ if and only if the value of the former is no less than that of the latter. In other words, $\tilde u$ ranks policy--state pairs $(\pi,P)$ via $P^\pi$. 

As a special case, we can define preferences over (non-oracle) policies within a fixed state of nature $P \in \Delta(\cY^{\cD}\times\cX)$, i.e.,
\begin{align}\label{eqn:pref_within_state}
    \pi \succsim_{(\tilde u,P)} \rho \iff P^\pi \succsim_{\tilde u} P^\rho.
\end{align}
However, allowing $P$ and $Q$ to differ is useful when comparing policies across distinct populations with different covariate and causal effect distributions. For example, one may wish to compare the implementation of alternative policies across cities when differences in their populations may warrant varying policy choices. If the budget allows the policy to be implemented in only one city, then the decision maker is forced to choose between different states of nature, $P$ and $Q$. Relevant discussions can be found in the literature on external validity and data fusion \citep[e.g.,][]{doi:10.1073/pnas.1510507113,Egami_Hartman_2023}.

Although some comparisons across states are immediate (e.g., a decision maker may prefer a healthier population over a sicker one), preferences across states reveal additional information about the utility function. If a preference relation is defined within a fixed state, many utility-state pairs $(\tilde u, P)$ can rationalize the same ordering in Equation~\eqref{eqn:pref_within_state}. Indeed, if we treat both utility and beliefs as subjective, they may become confounded; a high value may reflect either a high perceived likelihood of certain states or a high utility assigned to them \citep[e.g.,][]{karni1983state,AumannSavage71,rubin1987weak}.  By requiring the decision maker to rank elements in $\Delta(\cZ)$, we can recover a unique counterfactual utility $\tilde u$ (up to a positive affine transformation) that represents these preferences as shown in Theorem~\ref{thm:vnm} below. This yields not only an induced ordering of decisions, but also highlights the decision maker's tradeoffs across decisions and states. We provide further details on within-state comparisons in Appendix~\ref{app:fixed-P}.

In Appendix~\ref{app:unit_state}, we also study an alternative setup, in which the state of nature is taken to be the unit-level stratum and covariates, rather than the population from which the unit is sampled. 
We show that this makes the decision maker Bayesian, with a prior over unit strata. 
Moreover, even if all preference comparisons are observed, utility is identified only up to stratum-dependent additive constants.

We emphasize that counterfactual utility ranks distributions on the potential outcome space $\Delta(\cZ)$, rather than labels $d \in \cD$ themselves. This is consistent with standard vNM decision theory, where decisions carry no intrinsic value apart from the distribution of realized outcomes they induce.

\subsection{Axioms}\label{sec:axioms}

Expected counterfactual utility represents expected utility on the potential outcome space $\cZ$. Hence, the induced preferences are characterized by the von Neumann--Morgenstern (vNM) axioms \citep{NeumannMorgenstern1944}. We review these axioms and show how they apply to counterfactual utility.

Let $\succsim$ denote a preference relation on $\Delta(\cZ)$. For $p,q \in \Delta(\cZ)$, write $p \succ q$ if $p \succsim q$ but not $q \succsim p$, and write $p \sim q$ if both $p \succsim q$ and $q \succsim p$.
\begin{axiom}[Completeness]\label{axm:completeness}
For all $p, q \in \Delta(\cZ)$, either $p \succsim q$ or $q \succsim p$.
\end{axiom}
Axiom~\ref{axm:completeness} requires that one can compare any pair of distributions. In our setting, this means that we can rank any two policy--state pairs, $(\pi,P)$ and $(\rho,Q)$.

\begin{axiom}[Transitivity]\label{axm:transitivity}
For all $p, q, r \in \Delta(\cZ)$, if $p \succsim q$ and $q \succsim r$, then $p \succsim r$.
\end{axiom}
Axiom~\ref{axm:transitivity} requires that preferences are consistent across multiple options and rules out preference cycles. In our context, if the policy--state pair $(\pi,P)$ is preferred to $(\rho,Q)$ and $(\rho,Q)$ is preferred to $(\sigma,R)$, then $(\pi,P)$ must be preferred to $(\sigma,R)$. Together with Axiom~\ref{axm:completeness}, this axiom implies that $\succsim$ is a weak order.

For the next set of axioms, fix $\alpha \in [0,1]$. For $p,q \in \Delta(\cZ)$, define their convex combination by $\{\alpha p + (1-\alpha) q\}(z) := \alpha p(z) + (1-\alpha) q(z)$, $z \in \cZ$. We can think of this as drawing $Z$ from $p$ with probability $\alpha$ and from $q$ with probability $1-\alpha$, using this mixture to obtain an element of $\Delta(\cZ)$.

\begin{axiom}[Independence]\label{axm:independence}
For all $p, q,  r \in \Delta(\cZ)$ and $\alpha \in (0, 1]$, if $p \succ q$ then $\alpha p + (1 - \alpha) r \succ \alpha q + (1 - \alpha) r$.
\end{axiom}
Axiom~\ref{axm:independence} requires that strict preferences are preserved under mixing with a common third option. In our setting, $(\pi,P)$ is strictly preferred to $(\rho,Q)$, then mixing $(\pi,P)$ and $(\sigma,R)$ with with probabilities $\alpha$ and $1-\alpha$, respectively, is strictly preferred to the mixture that selects $(\rho,Q)$ with probability $\alpha$ and $(\sigma,R)$ with probability $1-\alpha$.

\begin{axiom}[Continuity]\label{axm:continuity}
For all $p, q,  r \in \Delta(\cZ)$, if $p \succ q \succ r$ then there exists $\alpha, \beta \in (0, 1)$ such that $\alpha p + (1 - \alpha) r \succ q \succ \beta p + (1-\beta) r$.
\end{axiom}
Axiom~\ref{axm:continuity} assumes that preferences are continuous so that they can be separated by sufficiently small deviations in probability.

The following theorem is an immediate application of the vNM theorem, applied to the potential outcome space $\cZ$ rather than the space of realized outcome $Y = Y(D) \in \cY$. It implies that preferences defined over both realized and counterfactual outcomes can be coherent. Modern proofs can be found in \cite{kreps1988notes} and \cite{Gilboa}.
\begin{theorem}[von Neumann and Morgenstern for Counterfactual Utility]\label{thm:vnm}
The following are equivalent:
\begin{enumerate}
    \item $\succsim$ satisfies Axioms~\ref{axm:completeness}--\ref{axm:continuity};
    \item there exists a counterfactual utility $\tilde u:\cZ\to\R$ that represents $\succsim$ in the sense of Definition~\ref{def:repr}.
\end{enumerate}
Moreover, any such utility function $\tilde u$ is unique up to a positive affine transformation.
\end{theorem}

\begin{remark}[Realized vs. Potential Outcome Spaces]\label{rem:realized_potential_space}
  At first glance, Theorem~\ref{thm:vnm} appears to contradict the following two claims in the literature: (1) counterfactual utility can result in nonsensical decision
  recommendations such as money pumps because ``it violates the axioms
  of classical (von Neumann-Morgenstern) decision theory''
  \citep[page~2]{Gelman_2025}, and (2) counterfactual utility can be intransitive
  \citep{sawant2026counterfactualharm}. This apparent
  conflict is resolved by distinguishing the space on which
  preferences are defined. As formally shown in
  Section~\ref{sec:mapping_counterfactual_preferences_to_choice}, the
  violation of vNM axioms can occur when one interprets the induced
  preferences over the realized outcome space rather than the potential outcome space.

By contrast, the counterfactual decision theory framework considered
here and elsewhere \citep[e.g.,][]{li2024unit,
  koch2025statisticaldecisiontheorycounterfactual} defines preferences
on the \emph{potential outcome space} $\cZ$.  That is, decisions are
evaluated with respect to both realized and counterfactual outcomes
$\{Y(k)\}_{k\in\cD}$. Theorem~\ref{thm:vnm}  implies that on this extended space, preferences induced by expected counterfactual utility satisfy the vNM axioms.
\end{remark}

\subsection{Relation to Standard Decision Theory}

Theorem~\ref{thm:vnm} shows that preferences induced by counterfactual utilities satisfy the generalized vNM axioms on the potential outcome space.  We now connect this result to the standard decision theory. Specifically, we characterize preferences induced by standard utilities, which depend only on $(D,Y(D), \bX) \in \cD \times \cY \times \cX$, when viewed as preferences on the potential outcome space $\cZ$. We show that such utilities assume indifference to the {\it counterfactual} outcomes under alternative decisions.  We introduce an axiom that formally characterizes the irrelevance of counterfactual outcomes.

\begin{axiom}[Irrelevance of Counterfactual Outcomes]\label{axm:irrelevance_of_untaken_decisions}
    Let $\delta_z(\cdot)$ denote the Dirac measure at $z$.  For all $p, q \in \Delta(\cZ)$, $d \in \cD$ and $\bx \in \cX$, if
    \begin{enumerate}
        \item $p(D = k) = q(D = k) = \delta_d(k), \quad \forall k \in \cD$,
        \item $p(\bX = \bx') = q(\bX = \bx') = \delta_{\bx}(\bx'), \quad \forall \bx' \in \cX$,
        \item $ p(Y(d) = y) = q(Y(d) = y), \quad \forall y \in \cY$,
    \end{enumerate}
    then $p \sim q$.
\end{axiom}

In our setting, Axiom~\ref{axm:irrelevance_of_untaken_decisions}
implies that a decision maker is indifferent between $(\pi,P)$ and
$(\rho,Q)$ if (1) the two policies always choose the same action $d$,
(2) both states are degenerate at the same covariate value $\bX=\bx$
such that all units have the same covariate values, and (3) the induced marginal distribution of the realized outcome $Y(d)$ coincides. Differences in the distribution of counterfactual outcomes $\{Y(k)\}_{k\neq d}$ are irrelevant under this axiom.
Axiom~\ref{axm:irrelevance_of_untaken_decisions} is weaker than requiring indifference between distributions that have the same joint distribution $(D, Y(D), \bX)$, although, under Axioms~\ref{axm:completeness}–\ref{axm:continuity}, the two conditions are equivalent.

The next result shows that Axioms~\ref{axm:completeness}--\ref{axm:irrelevance_of_untaken_decisions} characterize the preferences induced by standard utilities on the potential outcome space $\cZ$.
\begin{theorem}[Preferences of Standard Utilities on the potential outcome Space]\label{thm:std}
The following are equivalent:
\begin{enumerate}
    \item $\succsim$ satisfies Axioms~\ref{axm:completeness}--\ref{axm:irrelevance_of_untaken_decisions};
    \item there exists a standard utility $u(d; \by, \bx) = u(d; y_d, \bx)$ that represents $\succsim$ in the sense of Definition~\ref{def:repr}.
\end{enumerate}
Moreover, any such utility function $u$ is unique up to a positive affine transformation.
\end{theorem}

The proof is given in Appendix~\ref{app:std}.
Note that in standard decision theory, a utility often depends only on the realized outcome $Y(D)\in\cY$ (but not the decision), i.e., $\tilde u(d;\by,\bx)=u(y_d)$.
Appendix~\ref{app:outcome_utility} shows that this 
can be axiomatized using the potential outcome space $\cZ$ by strengthening
Axiom~\ref{axm:irrelevance_of_untaken_decisions}. In particular, such a utility implies indifference whenever the distribution of the realized outcome is identical, even if it arises under different decisions. The more general form above allows the decision label $d$ to matter, for example, to accommodate costs associated with the decision.

\subsection{Revisiting the Allais Paradox} \label{sec:revisit_examples_allais}

Armed with the formal axiomatic properties derived above, we now
revisit the Allais paradox described in Section~\ref{sec:examples}. 
The Allais paradox violates the independence axiom when preferences are defined on the realized outcome space $\cY$. Indeed, the first experiment corresponds to the distributions $p_{a_1} = 0.8 \delta_{4000} + 0.2 \delta_{0}$ and $p_{b_1} = \delta_{3000}$, whereas the second experiment corresponds to $p_{a_2} = 0.2 \delta_{4000} + 0.8 \delta_0$ and $p_{b_2} = 0.25 \delta_{3000} + 0.75 \delta_{0}$. Note that
$$
    p_{a_2} \ = \ 0.25 p_{a_1} + 0.75 \delta_0, \quad
    p_{b_2} \ = \ 0.25 p_{b_1} + 0.75 \delta_0.
$$
According to Axiom~\ref{axm:independence} (if restricted to $\Delta(\cY)$), $p_{b_1} \succ p_{a_1}$ would imply $p_{b_2} \succ p_{a_2}$. Yet, the experiment shows the opposite, suggesting that subjects' behavior violates  independence over $\Delta(\cY)$.

However, as we now show, this does not violate independence over $\Delta(\cZ)$. Let $\cD=\{0,1\}$ and $\cY$ be finite such that $\{0,3000,4000\} \subseteq \cY$. Further, let $\delta_{(y_0,y_1)}$ denote the Dirac measure at $(y_0,y_1)\in\cY^2$. Associate $a_1, a_2$ with $d=0$ and $b_1, b_2$ with $d=1$. Under the independent coupling of Assumption~\ref{ass:independence} (other couplings are also possible), the two states of nature over $(Y(0),Y(1))$ are given by,
\[
    P_1 := P_{a_1 b_1} = 0.8\,\delta_{(4000,3000)} + 0.2\,\delta_{(0,3000)},
\]
for the first experiment, and
\[
    P_2 := P_{a_2 b_2}
    = 0.05\,\delta_{(4000,3000)} + 0.15\,\delta_{(4000,0)}
    + 0.2\,\delta_{(0,3000)} + 0.6\,\delta_{(0,0)},
\]
for the second experiment. The induced measures on $\Delta(\cD\times\cY^2)$ are
$q_{a_1}=P_1^0 = \delta_{(d = 0)} P_1$, $q_{b_1}=P_{1}^1 = \delta_{(d = 1)} P_{1}$, $q_{a_2}=P_{2}^0 = \delta_{(d = 0)} P_{2}$, and $q_{b_2}=P_{2}^1 = \delta_{(d = 1)} P_{2}$.
Note that
\begin{align*}
    q_{a_2} &= \frac{1}{16} q_{a_1} + \frac{15}{16}\lt(0.16 \delta_{(d = 0)} \delta_{(4000, 0)} + 0.2 \delta_{(d = 0)} \delta_{(0, 3000)} + 0.64 \delta_{(d = 0)} \delta_{(0, 0)}\rt),
    \\
    q_{b_2} &= \frac{1}{16} q_{b_1} + \frac{15}{16} \lt(0.16 \delta_{(d = 1)} \delta_{(4000, 0)} + 0.2 \delta_{(d=1)} \delta_{(0, 3000)}+ 0.64  \delta_{(d = 1)} \delta_{(0, 0)} \rt).
\end{align*}
Unlike in the standard case, the remainder terms are not identical because $\delta_{(d=0)}\neq\delta_{(d=1)}$. Therefore, Axiom~\ref{axm:independence} (when applied to $\Delta(\cZ)$) does not imply $q_{b_2} \succ q_{a_2}$.

Under the utility function $\tilde u_{\mathrm{Bell}}$ introduced in Equation~\eqref{eqn:bell_utility}, for any $\lambda \ge 0.003$, the induced policy values satisfy the following inequalities,
\[
    V_{P_{1}}(1;\tilde u_{\mathrm{Bell}}) > V_{P_{1}}(0;\tilde u_{\mathrm{Bell}}) > V_{P_{2}}(0;\tilde u_{\mathrm{Bell}}) >  V_{P_{2}}(1;\tilde u_{\mathrm{Bell}}),
\]
implying $q_{b_1} \succ q_{a_1} \succ q_{a_2} \succ q_{b_2}$, which matches the observed choice pattern in the Allais paradox. This ranking indicates a preference of choosing $b_1$ in state $P_{1}$ to choosing $a_1$ in $P_{1}$, to choosing $a_2$ in $P_{2}$, to choosing $b_2$ in $P_{2}$. In particular, each choice is evaluated conditional on the relevant correlation structure encoded by the state.

In the presentation above, subjects face binary menus, so we model the experiment as two separate binary decisions with $\cD=\{0,1\}$. 
One can alternatively model the experiment as a single choice with $\cD=\{0,1,2,3\}$, comparing all lotteries simultaneously. 
Under this interpretation, the paradox can also be resolved, for example with
\begin{align}\label{eqn:bell_utility_extended}
    \tilde u(d; y_0, y_1, y_2, y_3) = y_d + \sum_{d^\prime \neq d} f_\lambda(y_d - y_{d^\prime}),
\end{align}
where $f_\lambda$ is defined in Equation~\eqref{eqn:bell_utility}. Under Assumption~\ref{ass:independence} (with independence among $\{Y(d)\}_{d \in \cD}$), for $\lambda \ge 0.002$, the induced ranking on $\Delta(\cD \times \cY^4)$ again matches the observed pattern with $b_1 \succ a_1 \succ a_2 \succ b_2$ given $P_{a_1 b_1 a_2 b_2}$ consistent with the marginals.

\section{Identification of Counterfactual Utilities} \label{sec:identification}

A central criticism of counterfactual utilities is that their expectations are typically not point identified from the observed data \citep[e.g.,][]{dawid2000causal, dawid2023personalised, sarvet2023perspectives, sarvet2024rejoinderperspectivesharmpersonalized}. While in general $V_P(\pi;\tilde u)$ depends on the joint distribution of the full potential outcome vector, we only observe one potential outcome per unit. We discuss two remedies considered in the literature: partial identification, which derives sharp bounds, and additive counterfactual utilities, whose expectations are point-identifiable. We briefly summarize the former approach and then focus on the latter by characterizing the axioms that are implied by additivity. This result directly links preferences on the potential-outcome space to point identification of expected counterfactual utility.

\subsection{Partial Identification}

Under a partial identification approach, one derives sharp bounds on the unobserved joint distribution of potential outcomes consistent with the observed marginals, and then evaluates each policy by its worst-case value over the resulting identified set \citep{ben-michaelPolicyLearningAsymmetric}.
Imposing additional assumptions (e.g., monotonicity) or leveraging external observational data, can tighten these bounds \citep{li2019unit, li2024unit}.
While always applicable, partial identification can be conservative and may fail to deliver a policy recommendation, a difficulty that is often exacerbated when there are more treatments under consideration.

\subsection{Additive Counterfactual Utilities}\label{sec:add}

An alternative to partial identification is to consider a class of counterfactual utilities whose expectation is identifiable. \citet{koch2025statisticaldecisiontheorycounterfactual} proves that $V_P(\pi, \tilde u)$ is identifiable for non-oracle policies for every state of nature $P$ if and only if $\tilde u$ is additive in potential outcomes, i.e., there exist $u_k: \cD \times \cY \times \cX \to \R$ for $k \in \cD$ such that
\begin{align*}
    \tilde u(d; \by, \bx) = \sum_{k \in \cD} u_k(d; y_k, \bx).
\end{align*}
Due to the additive structure, $V_P(\pi;\tilde u)$ depends only on the marginal distributions of the potential outcomes, and is therefore identifiable under the unconfoundedness of the decision. No other class of utility functions yields an identifiable expectation for every state of nature.

Additivity breaks the dependence structure of potential outcomes and
is thus less expressive (or less general) than fully counterfactual utilities. Unlike
standard utilities, however, it still incorporates counterfactual
outcomes. The next axiom formalizes this point. 
\begin{axiom}[Irrelevance of Counterfactual Correlation]\label{axm:indifference_to_correlation}
    For all $p, q \in \Delta(\cZ)$, $d \in \cD$ and $\bx \in \cX$, if
    \begin{enumerate}
        \item $p(D = k) = q(D = k) = \delta_d(k), \quad \forall k \in \cD$,
        \item $p(\bX = \bx') = q(\bX = \bx') = \delta_{\bx}(\bx'), \quad \forall \bx' \in \cX$,
        \item $p(Y(k) = y) = q(Y(k) = y), \quad \forall k \in \cD, \forall y \in \cY$,
    \end{enumerate}
    then $p \sim q$.
\end{axiom}

Axiom~\ref{axm:indifference_to_correlation} weakens Axiom~\ref{axm:irrelevance_of_untaken_decisions} by strengthening its premise. For the decision maker to be indifferent between distributions, the marginal distributions of \emph{all} potential outcomes, realized or counterfactual, must coincide. In other words, the decision maker cares about counterfactual outcomes, but is indifferent to their dependence structure.

The next theorem establishes that Axioms~\ref{axm:completeness}--\ref{axm:continuity} and~\ref{axm:indifference_to_correlation} characterize the preferences induced by additive utilities on the potential outcome space $\cZ$.  It implies that whenever a decision maker is additionally indifferent to counterfactual correlation, their preferences admit a representation by a point-identified expected counterfactual utility. The proof is given in Appendix~\ref{app:add}.
\begin{theorem}[Preferences of Additive Utilities on the potential outcome Space]\label{thm:add}
The following are equivalent:
\begin{enumerate}
    \item $\succsim$ satisfies Axioms~\ref{axm:completeness}--\ref{axm:continuity} and~\ref{axm:indifference_to_correlation};
    \item there exists an additive utility $\tilde u(d; \by, \bx) = \sum_{k \in \cD} u_k(d; y_k, \bx)$ that represents $\succsim$ in the sense of Definition~\ref{def:repr}.
\end{enumerate}
Moreover, the utility function $\tilde u$ is unique up to a positive affine transformation.
\end{theorem}

Proposition~1 of \cite{koch2025statisticaldecisiontheorycounterfactual} shows that when the decision is binary (i.e., $K=2$), for every additive counterfactual utility $\tilde u$, there exist a standard utility $u$ and a decision-independent function $h$ such that
\[
    \tilde u(d;\by,\bx) = u(d;y_d,\bx) + h(\by,\bx).
\]
Consequently, for any fixed state of nature $P$, $\tilde u$ and $u$ induce the same ranking over policies $\pi$. However, they need not induce the same preference relation on $\Delta(\cZ)$, since $\E_P[h((Y(0),Y(1)),\bX)]$ may vary with $P$ and thus affects comparison across states. This means that even in the binary decision case and under Axioms~\ref{axm:completeness}--\ref{axm:continuity}, Axiom~\ref{axm:indifference_to_correlation} is strictly weaker than Axiom~\ref{axm:irrelevance_of_untaken_decisions}. Too see this, consider the following additive utility,
\[
    \tilde u(d; \by) = y_0 + y_1,
\]
where $h = \tilde u$ and $u = 0$ in the above formulation. Consider the two states $P((Y(0),Y(1))=(0,1)) = 1$ and $Q((Y(0),Y(1))=(0,0)) = 1$. Let $d = 0$, then $P^0 \succ_{\tilde u} Q^0$. However, since $P^0(Y(0)) = Q^0(Y(0))$, Axiom~\ref{axm:irrelevance_of_untaken_decisions} implies $P^0 \sim Q^0$.

Furthermore, \cite{koch2025statisticaldecisiontheorycounterfactual} shows that such equivalence does not hold when $K\ge 3$ (Proposition~2). In these cases, additive counterfactual utilities are strictly more expressive than standard utilities, and their decision-making cannot, in general, be replicated by any standard utility.

\section{Counterfactual Utilities with Stochastic Potential Outcomes}\label{sec:spo}

We now examine the recent proposal by \cite{Gelman_2025} (GM) that counterfactual utilities be defined in terms of stochastic, rather than deterministic, potential outcomes. As discussed in Remark~\ref{rem:realized_potential_space}, the original motivation for this proposal is no longer relevant if preferences are modeled over the potential outcome space. Nevertheless, we consider it on its own terms as an alternative framework.  

Below, we first explain why the use of stochastic potential outcomes does not alter any of the axiomatic results derived above.  We then show that the key aspect of GM's proposal is the use of {\it extended} counterfactual utilities. We establish that these extended counterfactual utilities are not unique, reduce to standard utilities under a set of mild assumptions, and can violate the vNM independence axiom. 

\subsection{Use of Stochastic Potential Outcomes}

Our analysis so far has considered potential outcomes drawn from a super-population. In the causal inference literature, potential outcomes are commonly assumed to be \emph{fixed} for a given unit, with randomness arising only from sampling \citep{neyman1923application, rubin1974estimating}. This contrasts with \emph{stochastic} potential outcomes, which have an additional source of randomness for each unit \citep{greenland1987interpretation, vanderweele2012stochastic}.
As GM correctly recognizes, both frameworks ``imply the same joint population distribution'' (page~5). Since  our axiomatic results only require the existence of probability distribution over $Z = (D, Y(0), \ldots, Y(K-1), \bX)$, they directly apply to counterfactual decision theory under stochastic potential outcomes. In other words, the use of stochastic potential outcomes does not alter the results derived in Sections~\ref{sec:axiomatizing_counterfactual_decisiont_theory}~and~\ref{sec:identification}.

\subsection{Extended Counterfactual Utilities}\label{sec:setup_soe}

GM's proposal not only uses stochastic potential outcomes but also depends on an extension of counterfactual utilities. Counterfactual utilities discussed so far are mappings from the the decision, the set of potential outcomes, and covariates, i.e.,
\[
\tilde u:\cD\times \cY^\cD \times \cX \to \bR.
\]
Given a state of nature $P\in \Delta(\cY^\cD \times \cX)$ decisions are evaluated by their expected utility $V_P(d; \tilde u)$.

Instead, GM proposes to assign the value $p_d(P) := \E_P[Y(d)]$ to each stochastic potential outcome and then evaluate a decision using an \emph{extended} counterfactual utility that aggregates these values, i.e., $\tilde u^\ext(d; p_0(P), \ldots, p_{K-1}(P))$ (GM does not discuss randomized policies or covariates). Note that even when the outcome is discrete, we have $p_d(P) \in \R$, implying the need to extend the domain of utility. Because this construction depends only on marginal means, the extended counterfactual utility $\tilde u^{\ext}$ is identifiable.

For example, GM proposes to extend the original utility in Equation~\eqref{eqn:gm_utility} to the following utility function,
\begin{align*}
    \tilde u^{\ext}(0; p_0, p_1) = \indicator{p_0 > p_1}, \quad \text{and} \quad
    \tilde u^{\ext}(1; p_0, p_1) = 0.5\, \indicator{p_0 < p_1}.
\end{align*}
As shown by GM, applying this extended counterfactual utility to the Russian roulette example
discussed in Section~\ref{sec:russian_roulette} yields the  recommendation that the lottery with a higher survival probability $R_{1/7}$ should be chosen,
\begin{align}\label{eqn:gm_utility_extension}
    \tilde u^{\ext}(1; p_0,p_1) - \tilde u^{\ext}(0; p_0,p_1)
    = 0.5\,\indicator{p_0<p_1} - \indicator{p_0>p_1}
    = 0.5 > 0.
\end{align}
Since only marginals enter Equation~\eqref{eqn:gm_utility_extension}, Assumption~\ref{ass:independence} is not required for this result.

We can generalize GM's proposal to the following mapping,
\[
\tilde u^{\ext}:\cD \times \left(\Delta(\cY)\right)^\cD \times \cX \to \R.
\]
Given a conditional distribution of potential outcomes $\zeta:\cX \to \left(\Delta(\cY)\right)^\cD$, we can compute, for each covariate value $\bx$, the marginal outcome distribution associated with each decision and input it to $\tilde u^{\ext}$. Then, given a distribution $\mu\in \Delta(\cX)$, we can evaluate a decision rule by integrating $\tilde u^{\ext}$ with respect to $\mu$.  We note that stochastic potential outcomes are not required for this procedure. Indeed, the same exact setup can be applied to deterministic potential outcomes, which have the same marginal distributions, yielding the identical recommendation.

In general, the expectation of $\tilde u^{\ext}$ will be different from that of $\tilde u$, even if $P=\mu\zeta$. In particular, given a fixed $\tilde u$, there exist many ways to define $\tilde u^{\ext}$ so that it coincides with $\tilde u$ on Dirac measures. Next, we will discuss this non-uniqueness problem.

\subsection{Non-Uniqueness}\label{sec:non_unique_extension}

As noted above, GM’s procedure requires extending the domain of the utility function from binary inputs $y_d \in \{0,1\}$ to real-valued quantities $p_d = \E[Y(d)] \in [0,1]$. However, this extension of counterfactual utility is not unique, and alternative extensions may yield different recommendations.

To illustrate this in the Russian roulette example, note that the utility contrast induced by Equation~\eqref{eqn:gm_utility} can be written in two ways,
\begin{align*}
   \tilde u(1; y_0, y_1) - \tilde u(0; y_0, y_1) = 0.5\, \indicator{y_0 < y_1} - \indicator{y_0 > y_1} = 0.5 \, (1-y_0) y_1 - y_0 (1-y_1).
\end{align*}
Although these representations are equivalent on the binary domain, they lead to different extensions and hence can yield opposite recommendations. Specifically, the first representation yields the {\it asymmetric} extension, recommending $R_{1/7}$, as shown in Equation~\eqref{eqn:gm_utility_extension}.
In contrast, the second representation leads to the following {\it product} extension, recommending $R_{1/6}$,
$$
\tilde u^{\ext}_{\mathrm{Prd}}(1; p_0,p_1) - \tilde u^{\ext}_{\mathrm{Prd}}(0; p_0,p_1) = 0.5\,(1-p_0)p_1 - p_0(1-p_1) = -\frac{1}{21}<0.
$$

In fact, there are infinitely many such extensions from $\tilde u:\{0,1\}^3\to\R$ to $\tilde u^\ext:\{0,1\}\times[0,1]^2\to\R$ that are equivalent on the original domain. Each extension implicitly encodes an additional structure. For instance, the product extension is equivalent to making the independence assumption (Assumption~\ref{ass:independence}) within the original framework of Section~\ref{sec:setup}, which is restrictive in many applications.

This issue becomes especially apparent when utilities are specified as payoff
tables rather than as closed-form expressions. In fact,
mathematically, a utility function $\tilde u:\cD\times \cY^\cD \times
\cX \to \bR$ can be represented by a vector of real numbers. For
example, the utility in Equation~\eqref{eqn:gm_utility} might only be
given in the following matrix form:
\[
\tilde u(0;\,y_0,y_1)
=
\lt[
\begin{array}{c|cc}
 & y_1=0 & y_1=1 \\ \hline
y_0=0 & 0 & 0 \\
y_0=1 & 1 & 0
\end{array}
\rt],
\qquad
\tilde u(1;\,y_0,y_1)
=
\lt[
\begin{array}{c|cc}
 & y_1=0 & y_1=1 \\ \hline
y_0=0 & 0 & \tfrac12 \\
y_0=1 & 0 & 0
\end{array}
\rt].
\]
GM's procedure cannot be applied to payoff tables directly. One must first choose a functional form representation consistent with the table. This step can be nontrivial and, as we have seen, is generally non-unique.


\subsection{Reduction to Standard Utility}\label{sec:failing_to_preserve_asymmetry}

The difference in the extended counterfactual utility given in
Equation~\eqref{eqn:gm_utility_extension} reduces to a comparison of
marginal means, i.e., whether $p_1 > p_0$. It therefore yields the
same recommendation as the standard utility $u(d; y_d) = y_d$ under
the framework introduced in Section~\ref{sec:setup}. In particular,
the comparison no longer depends on the joint distribution of
potential outcomes, thereby eliminating the asymmetry that is present in the original counterfactual utility. In Appendix~\ref{app:failing_to_preserve_asymmetry}, we formalize and further generalize this result to arbitrary discrete outcomes.  Specifically, according to Theorem~\ref{thm:gm_reduction_to_std} in the appendix, under mild conditions and in the case of binary decisions, for every extended counterfactual utility, there exists an extended standard utility that yields the same recommendation in every state of nature.

\subsection{Violation of the Independence Axiom}

We next show that the extended counterfactual utilities proposed by GM do not satisfy the generalized vNM independence axiom (Axiom~\ref{axm:independence}) over the potential outcome space. Indeed, preferences satisfy Axiom~\ref{axm:independence} only if $\tilde u^{\ext}$ admits a representation that is affine in the probability measure.
This result does not contradict the claim given in
Section~\ref{sec:failing_to_preserve_asymmetry} that, for a given
extended counterfactual utility $\tilde u^{\ext}$, there exists an
equivalent, extended standard utility that yields the same optimal
policy for any state of nature $P$. The reason is that here we focus on the preference relation over all (possibly non-optimal) decisions and allow comparisons across states.

We begin by noting that an extended counterfactual utility induces the following preference relation over $\Delta(\cZ)$,
\begin{align}\label{eqn:gm_pref_rel}
    P^d \succsim Q^k
    \iff
    \tilde u^{\ext}\!\left(d; p_0(P),\ldots,p_{K-1}(P)\right)
    \ge
    \tilde u^{\ext}\!\left(k; p_0(Q),\ldots,p_{K-1}(Q)\right).
\end{align}
Like GM, we focus on deterministic decisions.
Consider the extension given in Equation~\eqref{eqn:gm_utility_extension}. Let $P,Q \in \Delta(\cY^2)$ satisfy
\[
    P((Y(0),Y(1))=(0,1))=1,
    \qquad
    Q((Y(0),Y(1))=(1,0))=1,
\]
with marginal means $(p_0,p_1)=(0,1)$ and $(q_0,q_1)=(1,0)$. For the fixed decision $d = 1$,
\[
    \tilde u^{\ext}(1; p_0,p_1)=0.5,
    \qquad
    \tilde u^{\ext}(1; q_0,q_1)=0,
\]
and hence $P^1 \succ Q^1$.
Now, set $R=Q$ and consider the following mixtures with the same weight $\alpha=0.4$,
\[
    L^1 := \alpha P^1 + (1-\alpha)R^1 = 0.4 P^1 + 0.6 Q^1,
    \qquad
    M^1 := \alpha Q^1 + (1-\alpha)R^1 = Q^1,
\]
where marginal means are $(l_0,l_1)=(0.6,0.4)$ and $(m_0,m_1)=(1,0)$. Thus, the extended counterfactual utilities are given by,
\[
    \tilde u^{\ext}(1; l_0,l_1)=0,
    \qquad
    \tilde u^{\ext}(1; m_0,m_1)=0,
\]
and therefore $L^1 \sim M^1$. This contradicts Axiom~\ref{axm:independence}, which would require $L^1 \succ M^1$ because $P^1 \succ Q^1$ and $L^1$ and $M^1$ are mixtures of $P^1$ and $Q^1$ with the same $R^1$ and the same weight $\alpha$. 

Thus, the preference relation induced by extended counterfactual utility does not, in general, define a coherent vNM preference relation on the potential outcome space. Moreover, as we show in Section~\ref{sec:mapping_counterfactual_preferences_to_choice}, projecting these preferences onto the realized outcome space, with either the menu-dependent or context-dependent projection produces the same vNM violations as with classical counterfactual utility.

\subsection{Extended Counterfactual Utilities under Additivity}

We can recover coherence on the potential outcome space by imposing additivity on the extended utility. Assume that the extended utility takes the additive form,
\[
    \tilde u^{\ext}(d; p_0,\ldots,p_{K-1})
    = \alpha_d + \sum_{k\in\cD} \beta_{d,k}\,p_k,
\]
for constants $\alpha_d,\beta_{d,k}\in\R$. This can be viewed as a
natural extension of the additive counterfactual utility $ \tilde
u(d;\by)=\alpha_d+\sum_{k\in\cD}\beta_{d,k}\,y_k$ from Section~\ref{sec:add}, though this
transformation is not one-to-one. We have $\tilde u^{\ext}(d;
p_0(P),\ldots,p_{K-1}(P)) =  V_P(d;\tilde u)$ because additive
utilities also depend only on marginal distributions alone.  Thus, as
shown earlier, the preference relation induced by $\tilde u^{\ext}$ satisfies the vNM independence axiom.

We emphasize that this result does not contradict the conclusions of Sections~\ref{sec:add} and~\ref{sec:failing_to_preserve_asymmetry}. As shown above, under binary decisions, additivity reduces the extended counterfactual utilities to standard utilities.  However, when there are more than two decision categories, additive counterfactual utilities no longer reduce to standard utilities and are more expressive.

\section{Projecting Counterfactual Preferences onto Realized Outcome Space}\label{sec:mapping_counterfactual_preferences_to_choice}

In the previous sections, we characterized preferences induced by counterfactual utilities on
the potential outcome space $\cZ$. In behavioral economics and
psychology, preferences are often expressed over lotteries
on the \emph{realized} (monetary) outcome space $\cY$. Therefore, as
briefly mentioned in Remark~\ref{rem:realized_potential_space}, we consider how
to ``project'' preferences from $\Delta(\cZ)$ onto $\Delta(\cY)$ through a counterfactual utility. In what follows, we discuss two such projections, one that is menu-dependent and another that is context-dependent.

\subsection{Setup}

Let $\cD$ be a finite and nonempty index set, where
each $d \in \cD$ corresponds to a lottery $p_d \in \Delta(\cY)$ with $\cY \subseteq \R$. A
\emph{menu} is a nonempty, finite subset $\cA \subseteq \cD$ from
which an individual chooses. Note that the individual is presented
with marginal distributions of potential outcomes, but their correlation is not given. Choice
can be summarized by the following \emph{choice function}
\[
    \chi: 2^{\cD}\setminus\{\emptyset\} \to 2^{\cD}\setminus\{\emptyset\},
\]
with $\chi(\cA)\subseteq \cA$ for all $\cA\subseteq \cD$. For example, if an individual chooses lottery $p_a$ from menu $\cA=\{a,b\}$, then $\chi(\cA)=\{a\}$. For simplicity of exposition, we suppress the covariates, though it is possible for $\chi$ to depend on them.

A special role is played by choice functions that are induced by maximizing a \emph{value} function $V:\cD\to \R$:
\begin{equation}\label{eqn:ind-choice}
\chi(\cA)=\argmax_{d\in \cA} V(d).
\end{equation}
Such choice behavior can be characterized by two axioms \citep{sen1971choice}.

\begin{axiom}[Sen's $\alpha$ condition]\label{axm:sen_alpha}
    If $d\in \cA \subseteq \mathcal{B}$ and $d\in \chi(\mathcal{B})$, then $d\in \chi(\cA)$.
\end{axiom}
Axiom~\ref{axm:sen_alpha} is implied by value maximization in
Equation~\eqref{eqn:ind-choice}, since a choice $d$ selected from menu $\cB$ must have a higher value than any other element of $\cB$, and therefore any element of $\cA \subseteq \cB$. This axiom is also known as independence of irrelevant alternatives (IIA) or Chernoff's condition \citep{chernoff1954rational}.

The second condition is technical and concerns tie-breaking.
\begin{axiom}[Sen's $\beta$ condition]\label{axm:sen_beta}
	If $d, d^\prime \in \chi(\cA)$, $\cA \subseteq \cB$ and $d^\prime \in \chi(\cB)$, then $d\in \chi(\cB)$.
\end{axiom}

\begin{theorem}[Value Maximization]\label{thm:WARP} The following are equivalent:
	\begin{enumerate}
		\item $\chi$ satisfies Axioms~\ref{axm:sen_alpha}~and~\ref{axm:sen_beta};
		\item there exists a value function $V:\cD\to \bR$ that induces $\chi$ in the sense of Equation~\eqref{eqn:ind-choice}.
	\end{enumerate}
\end{theorem}
A proof can be found in \citet{kreps1988notes}. Axioms~\ref{axm:sen_alpha}~and~\ref{axm:sen_beta} together are known as the Weak Axiom of Revealed Preferences (WARP). Under WARP, the revealed preference relation
\[
d \succsim^{\textsc{R}} d^\prime \iff d \in \chi(\{d, d^\prime\})
\]
is complete and transitive and is represented by a value function $V$.

\subsection{Menu-Dependent Projection}\label{subsec:menu_projection}

We first consider the menu-dependent projection of counterfactual
preferences onto $\Delta(\cY)$. This projection can capture regret
behavior \citep{bell1982regret, loomes1982regret,
  lanzani2022correlation} and connects to the confusion about
counterfactual utility discussed in
Remark~\ref{rem:realized_potential_space}. This projection may not satisfy WARP and can induce intransitive revealed preferences.

For each menu $\cA \subseteq \cD$, write $Y(\cA):=(Y(k))_{k\in\cA}$ and choose a counterfactual utility $\tilde u_{\cA}:\cA\times\cY^{\cA}\to\R$, where
$\cY^{\cA}$ is the $\cA$-indexed outcome space. Although only the marginal lotteries are presented to the individual, evaluating a (non-additive) counterfactual utility requires a coupling. Thus, specify a state of nature $P_{\cA}\in\Delta(\cY^{\cA})$ such that each lottery $d\in\cA$ is a marginal of $P_{\cA}$, i.e., $P_{\cA}(Y(d)) = p_d$. Further, let $P_{\cA}^d := \delta_d \cdot P_{\cA} \in \Delta(\cA \times \cY^\cA)$ denote the induced distribution under decision $d \in \cA$.

Define
\[
    \chi^{\textsc{M}}(\cA):=\argmax_{d\in \cA} \E_{P_\cA}[\tilde u_{\cA}(d;Y(\cA))].
\]
Because $(\tilde u_{\cA},P_{\cA})$ may vary with $\cA$, $\chi^{\textsc{M}}$ can violate Sen's $\alpha$ condition (Axiom~\ref{axm:sen_alpha}). As a result, the induced pairwise revealed preference relation on $\Delta(\cY)$, which is defined as follows, may not be transitive,
\begin{align}\label{eqn:bin_rel_2}
    d \succsim^{\textsc{M}} d^\prime
    &\iff
    d \in \chi^{\textsc{M}}(\{d,d^\prime\})
    \notag
    \\
    &\iff P_{\{d, d^\prime\}}^d  \succsim_{\tilde u_{\{d, d^\prime\}}} P_{\{d, d^\prime\}}^{d^\prime}
    \notag
    \\
    &\iff
    \E_{P_{\{d, d^\prime\}}}\lt[\tilde u_{\{d, d^\prime\}}(d; Y(\{d, d^\prime\}))\rt]
    \geq
    \E_{P_{\{d, d^\prime\}}}\lt[\tilde u_{\{d, d^\prime\}}(d'; Y(\{d, d^\prime\}))\rt],
\end{align}
where $\succsim_{\tilde u_{\{d, d^\prime\}}}$ is the underlying preference on the potential-outcome space induced by the menu utility
$\tilde u_{\{d,d'\}}$, while $\succsim^{\textsc{M}}$ compares the indexed marginal lotteries after a coupling has been specified for each binary menu.

Moreover, intransitivity can persist even when the menu size is fixed (so a single, menu-independent utility can be chosen) and a common correlation structure is imposed in the joint probability measure (e.g., independence). The following example illustrates this point.

\begin{example}\label{exp:menu_dep}
    Let $\cY = \{1,2,3,4,5,6,7,8,9\}$ and $\{a, b, c\} \subseteq \cD$ such that
    \[
        p_a=\frac{1}{3}(\delta_2+\delta_4+\delta_9),\qquad
        p_b=\frac{1}{3}(\delta_1+\delta_6+\delta_8),\qquad
        p_c=\frac{1}{3}(\delta_3+\delta_5+\delta_7).
    \]
    For every binary menu $\cA \subseteq \{a, b, c\}$, take the independent coupling $P_{\cA}=\bigotimes_{k\in\cA}p_k$ and let $\tilde u(d; (y_k)_{k \in \cA}) = \indicator{y_d > y_{d^\prime}}$ where $\{d^\prime\} =\cA \setminus \{d\}$. Then, we have,
    \[
        P(Y(a)>Y(b))=\frac{5}{9}>\frac{4}{9}=P(Y(b)>Y(a)),
        \]
        \[
        P(Y(b)>Y(c))=\frac{5}{9}>\frac{4}{9}=P(Y(c)>Y(b)),
        \]
        \[
        P(Y(c)>Y(a))=\frac{5}{9}>\frac{4}{9}=P(Y(a)>Y(c)),
    \]
    leading to,
    \[
    a\succ^{\textsc{M}} b,\qquad b\succ^{\textsc{M}} c,\qquad c\succ^{\textsc{M}} a.
    \]
    Thus, the pairwise revealed preference relation induced by the menu-dependent projection is not transitive. The same counterfactual utility is used in the intransitivity result of \citet{sawant2026counterfactualharm}, which we revisit in Example~\ref{exp:stensrud} below.

    However, the preference relation is transitive on the potential outcome space. Facing three separate binary choice problems we can model them with $\cD=\{0,1\}$, common utility $\tilde u$, and the same independence assumption. Order the pairs as $(a,b)$, $(b,c)$, and $(a,c)$, and write $P_{ab}$, $P_{bc}$, and $P_{ac}$ for the corresponding ordered laws, viewing the first and second alternatives as decisions $0$ and $1$.
Choosing $a$ over $b$ corresponds to $P_{ab}^a=\delta_{(d=0)}P_{ab} \succ \delta_{(d=1)}P_{ab}=P_{ab}^b$, choosing $b$ over $c$ to $P_{bc}^b=\delta_{(d=0)}P_{bc} \succ \delta_{(d=1)}P_{bc}=P_{bc}^c$, and choosing $c$ over $a$ to $P_{ac}^c=\delta_{(d=1)}P_{ac} \succ \delta_{(d=0)}P_{ac}=P_{ac}^a$.  Since each comparison is made within a different state, transitivity is not violated. In fact, we obtain the global ranking $P_{ab}^a \sim P_{bc}^b \sim P_{ac}^c \succ P_{ab}^b \sim P_{bc}^c \sim P_{ac}^a$.
\end{example}

\begin{example}
    The same setup as in Example~\ref{exp:menu_dep} also yields a violation of Axiom~\ref{axm:sen_alpha}. Let $\cA = \{a, c\}$ and $\cB = \{a, b, c\}$. For $\cM \in \{\cA, \cB\}$, take the independence coupling $P_{\cM}=\bigotimes_{k\in\cM}p_k$ and let $\tilde u_{\cM}(d; (y_k)_{k \in \cM}) = \sum_{k\in\cM\setminus\{d\}}\indicator{y_d>y_k} + \indicator{\cM=\cB}\indicator{d=a}$. 
    Then, $\chi^{\textsc{M}}(\cA) = \{c\}$, but $\chi^{\textsc{M}}(\cB) = \{a\}$. Since $a \in \cA$, this violates Sen's $\alpha$ condition.
\end{example}


We now examine when $\succsim^{\textsc{M}}$ is transitive. Fix a counterfactual utility $\tilde u:\{0,1\}\times\cY^2\to\R$ across all binary menus and assume $\tilde u(0;y_0,y_1) =  \tilde u(1;y_1,y_0)$ where $y_0,y_1\in\cY$. This condition ensures that the comparison is invariant to the arbitrary labeling of the two alternatives in a binary menu as $0$ and~$1$.

\begin{proposition}[Transitivity under Menu-dependent Projection]\label{prop:method2_transitive}
    Let $\cY$ be finite and $\tilde u: \{0, 1\} \times \cY^2 \to \R$ be a counterfactual utility with $\tilde u(0; y_0, y_1) = \tilde u(1; y_1, y_0), \, \forall y_0, y_1 \in \cY$. Then, the following are equivalent:
    \begin{enumerate}
        \item For every finite index set $\cD$, every family of marginal lotteries $(p_d)_{d \in \cD}$ and every assignment of couplings $P_{\{d, d^\prime\}} \in \Delta(\cY^{\{d, d^\prime\}})$ to binary menus $\{d, d^\prime\} \subseteq \cD$, such that $P_{\{d, d^\prime\}}(Y(k)) = p_k$ for $k \in \{d, d^\prime\}$, the preference relation $\succsim^{\textsc{M}}$, 
        \[
            d \succsim^{\textsc{M}} d^\prime \iff \E_{P_{\{d, d^\prime\}}}\lt[\tilde u(0; Y(d), Y(d^\prime))\rt]
    \geq
    \E_{P_{\{d, d^\prime\}}}\lt[\tilde u(1; Y(d), Y(d^\prime))\rt],
        \]
        is transitive.
        \item There exists $u: \cY \to \R$ and a symmetric function $h: \cY^2 \to \R$, i.e., $h(x, y) = h(y, x)$, such that
        \[
           \tilde u(d; y_0,y_1) = u(y_d) + h(y_0,y_1).
        \]
    \end{enumerate}
\end{proposition}

This result is a reformulation of \citet[Proposition~1]{lanzani2022correlation} and \citet[Theorem~1]{bikhchandani2011transitive}. There, the first condition is precisely the transitivity axiom for the preference set $\Pi_{\tilde u} = \{P \in \Delta(\cY^2): \E_P[\tilde u(0; y_0, y_1)] \geq \E_P[\tilde u(1; y_0, y_1)]\}$, i.e.,   for $P,Q,R\in\Delta(\cY^2)$ with marginals $P_0=R_0, P_1=Q_0, Q_1=R_1$, if $ P,Q\in\Pi_{\tilde u}$ then $R\in\Pi_{\tilde u}$.
Proposition~\ref{prop:method2_transitive} shows that, under the menu-dependent projection, transitivity can be recovered only by restricting to utilities that effectively ignore correlation across lotteries. This contrasts with counterfactual decision making on the potential outcome space, which remains transitive even under correlation because the dependence structure is part of the state being evaluated. 

For binary menus, the menu-dependent projection subsumes expected-regret models \citep[e.g.,][]{bell1982regret, loomes1982regret}, since it allows arbitrary utility--state pairs. It is closely related to \citet{lanzani2022correlation}, though he works with utility contrasts and preferences on $\Delta(\cY^2)$. For non-binary menus, however, the menu-dependent projection does not pin down a unique way to rank suboptimal choices within a menu.

\subsection{Context-Dependent Projection}\label{subsec:context_projection}

We next consider the menu-independent but context-dependent projection that incorporates the full set of available lotteries directly into the utility, inducing choices satisfying WARP.

Let $\cD = \{0, 1, \ldots, K-1\}$ index the set of available lotteries. Fix a counterfactual utility $\tilde u:\cD\times\cY^{\cD}\to\R$ and a state of nature $P\in\Delta(\cY^{\cD})$ such that each lottery $p_d\in\Delta(\cY)$ is a marginal of $P$, i.e., $P(Y(d))=p_d$ for all $d \in \cD$. Define the choice function,
\[
    \chi^{\textsc{C}}_{(\tilde u, P)}(\cA):=\argmax_{d\in \cA} \E_P[\tilde u(d; Y(0), \ldots, Y(K-1))].
\]
Its associated value function is $V_P(d;\tilde u)=\E_P[\tilde u(d; Y(0),\ldots,Y(K-1))]$. Hence, by Theorem~\ref{thm:WARP}, $\chi^{\textsc{C}}_{(\tilde u, P)}$ satisfies Sen's $\alpha$ and $\beta$ conditions (Axioms~\ref{axm:sen_alpha}~and~\ref{axm:sen_beta}). The induced revealed preference relation, which is given below, is therefore complete and transitive,
\[
    d \succsim^{\textsc{C}}_{(\tilde u,P)} d'
    \iff
    d \in \chi^{\textsc{C}}_{(\tilde u,P)}(\{d,d'\}).
\]
However, this preference relation might still be context-dependent because a preference between $d$ and $d'$ may depend on what other lotteries are in $\cD$. The following example highlights this point.

\begin{example}\label{exp:context_dep}
    Let $\cY=\{0,1\}$ and consider the lotteries
    \[
        p_a=\delta_{1},\qquad
        p_b=\frac12\delta_{0}+\frac12\delta_{1},\qquad
        p_{c_1}=\frac34\delta_{0}+\frac14\delta_{1},\qquad
        p_{c_2}=\frac14\delta_{0}+\frac34\delta_{1}.
    \]
    Let $\cD_1=\{a,b,{c_1}\}$ and $\cD_2=\{a,b,{c_2}\}$, and consider the same menu $\cA=\{a,b\}$ in both cases. Define an additive counterfactual utility,
    \[
        \tilde u(a; y_a,y_b,y_c)= 1 - y_c,
        \qquad
        \tilde u(b; y_a,y_b,y_c)=y_b.
    \]
    Then, we have,
    \[
        \E[\tilde u(a;Y(a),Y(b),Y(c))] = P(Y(c)=0),
        \qquad
        \E[\tilde u(b;Y(a),Y(b),Y(c))] = P(Y(b) = 1) =\frac12.
    \]
    Under $\cD_1$, $P(Y(c_1)=0)=3/4$, so $a \succ^{\textsc{C}}_{(\tilde u,P)} b$. Under $\cD_2$, $P(Y(c_2)=0)=1/4$, so $b \succ^{\textsc{C}}_{(\tilde u,P)} a$. As this counterfactual utility is additive, the conclusion does not depend on a particular choice of coupling. However, if we modify the additive counterfactual utility to
    \[
        \bar u(a; y_a,y_b,y_c)= 2 + y_c,
        \qquad
        \bar u(b; y_a,y_b,y_c)=y_b,
    \]
    then, we have $\bar u(a; y_a,y_b,y_c) - \bar u(b; y_a,y_b,y_c) \geq 1$ for all $\by \in \cY^3$. Thus, $a \succ^{\textsc{C}}_{(\bar u,P)} b$ for every $P$ and $\succsim^{\textsc{C}}_{(\bar u,P)}$ is context-independent for the pair $(a, b)$.
\end{example}

While standard utilities $u(d;y_d)$ are inherently context-independent, the example above shows that additive counterfactual utilities may or may not be, depending on their structure. Thus, standard utilities are sufficient for context independence, but not necessary, whereas additive counterfactual utilities are not sufficient.

Moreover, the induced relation $\succsim^{\textsc{C}}_{(\tilde u,P)}$ generally does not satisfy the vNM independence (Axiom~\ref{axm:independence}) or continuity (Axiom~\ref{axm:continuity}) on $\Delta(\cY)$. The next example demonstrates the violation of vNM independence.

\begin{example}
Assume the setup of Example~\ref{exp:context_dep} and focus on $\cD_1=\{a,b,c_1\}$. Extend $\tilde u$ by
\[
\tilde u(c_1;y_a, y_b, y_{c_1})= \frac{5}{8}.
\]
Then, $V_P(c_1;\tilde u)=5/8$, such that
\[
    a \succ^{\textsc C}_{(\tilde u,P)} c_1 \succ^{\textsc C}_{(\tilde u,P)} b,
\]
for any $P$ consistent with the marginals.
Now take $\alpha=1/3$. If the induced preference on $\Delta(\cY)$ satisfied vNM independence, $a \succ^{\textsc C}_{(\tilde u,P)} c_1$ would imply
\[
    p_b = \frac{1}{3} p_a + \frac{2}{3} p_{c_1} \succ^{\textsc C}_{(\tilde u,P)} \frac{1}{3} p_{c_1} + \frac{2}{3} p_{c_1} = p_{c_1},
\]
contradicting $c_1 \succ^{\textsc C}_{(\tilde u,P)} b$.

However, the same example does not violate independence on the potential outcome space. To see this, associate $\cD$ with $\{0, 1, 2\}$ and fix a state $P$ consistent with the marginals. Then $P^a = \delta_{(d = 0)}P, P^b = \delta_{(d = 1)}P, P^{c_1} =  \delta_{(d = 2)}P$. Now $\frac{1}{3}P^a + \frac{2}{3}P^{c_1}  = \frac{1}{3} \delta_{(d = 0)}P+ \frac{2}{3}\delta_{(d = 2)}P \neq \delta_{(d = 1)}P = P^b$ due to the difference in the decision indicator. Thus independence does not force a preference reversal.
\end{example}


\subsection{Discussion}

Each projection models preferences on $\Delta(\cY)$ differently. The menu-dependent projection assumes that only lotteries available in the current menu affect choice. In contrast, the context-dependent projection assumes that choice depends on all lotteries, including those not currently available.

For example, in the Allais paradox of Section~\ref{sec:revisit_examples_allais}, both methods take
$\cD=\{a_1,b_1,a_2,b_2\}$. The menu-dependent projection models the data as two binary menus, $\cA_1=\{a_1,b_1\}$ and $\cA_2=\{a_2, b_2\}$, and may choose the utility given in Equation~\eqref{eqn:bell_utility}. On the potential outcome space, this corresponds to treating each menu as a separate binary decision problem with $\cD=\{0,1\}$, associating $a_1,a_2$ with $d=0$ and $b_1, b_2$ with $d=1$, as in Section~\ref{sec:revisit_examples_allais}.

In contrast, the context-dependent projection might use the same menus but evaluates choices directly on $\cD=\{a_1, b_1, a_2, b_2\}$ and may use the utility given in Equation~\eqref{eqn:bell_utility_extended}.  On the potential-outcome space, this corresponds to modeling the setting as a single four-option decision problem.

As shown above, both projections are compatible with expected counterfactual utility and yield transitive preferences on the potential outcome space. However, on the marginal space $\Delta(\cY)$, only the context-dependent projection is always transitive (more precisely, it is transitive on the subset of $\Delta(\cY)$ that is indexed by $\cD$).

The next example, taken from \citet[Section~3]{sawant2026counterfactualharm}, illustrates this further.

\begin{example}\label{exp:stensrud}
    \citet{sawant2026counterfactualharm} consider a setting, in which a decision maker chooses among three medical treatments $\{a,b,c\}$ and observes an ordered health outcome in $\cY=\{1,2,3,4,5,6\}$ (greater means better health). The authors are interested in preferences on the realized outcome space $\Delta(\cY)$ and wish to select the single best option from $\{a,b,c\}$. A natural approach is the context-dependent projection introduced in Section~\ref{subsec:context_projection} that satisfies transitivity. This corresponds to specifying a utility $\tilde u: \cD \times \cY^3 \to \R$ and choosing according to $$\chi^{\textsc{C}}_{(\tilde u, P)}(\{a,b,c\})=\argmax_{d\in \{a,b,c\}} \E_P[\tilde u(d; Y(a), Y(b), Y(c))].$$

    Instead, the authors model the setting as facing three separate binary choice problems, along the lines of Section~\ref{subsec:menu_projection}. They motivate this approach by viewing the decision as a choice among three (binary) randomized controlled trials. They consider the utility $\tilde u(d; y_0, y_1) = \indicator{y_d > y_{1- d}}$ with the following marginal distributions
    \[
        p_a=\frac{1}{6} \delta_{1} + \frac{5}{6} \delta_{4} ,\qquad
        p_b=\frac{1}{2} \delta_{2} + \frac{1}{2} \delta_{5} ,\qquad
        p_c=\frac{5}{6} \delta_{3} + \frac{1}{6} \delta_{6}.
    \]
    Under this setup, the authors show that under any coupling $P \in \Delta(\cY^3)$ 
    \[
        P(Y(b)>Y(a)) \geq P(Y(a)>Y(b)),
        \]
        \[
        P(Y(a)>Y(c)) > P(Y(c)>Y(a)),
        \]
        \[
        P(Y(c)>Y(b)) \geq P(Y(b)>Y(c)),
    \]
    which implies, 
    \[
    b\succsim^{\textsc{M}} a,\qquad a\succ^{\textsc{M}} c, \qquad c\succsim^{\textsc{M}} b.
    \]
    Thus, this pairwise revealed preference relation induced by the menu-dependent projection is not transitive.
    
    However, the preference relation is transitive on the potential outcome space. Following \citet{sawant2026counterfactualharm}, consider three separate binary choice problems.  For any state of nature $P$ with binary joints $P_{ab}$, $P_{bc}$, and $P_{ac}$, and utility $\tilde u$, choosing $b$ over $a$ corresponds to $P_{ab}^b=\delta_{(d=1)}P_{ab} \succsim \delta_{(d=0)}P_{ab}=P_{ab}^a$, choosing $a$ over $c$ to $P_{ac}^a=\delta_{(d=0)}P_{ac} \succ \delta_{(d=1)}P_{ac}=P_{ac}^c$ and
    choosing $c$ over $b$ to $P_{bc}^c=\delta_{(d=1)}P_{bc} \succsim \delta_{(d=0)}P_{bc}=P_{bc}^b$.
    Since each comparison is made within a different state, transitivity is not violated. 
\end{example}

\section{Concluding Remarks}

In this paper, we establish an axiomatic foundation for counterfactual statistical decision theory, which has recently gained popularity but attracted criticisms at the same time.  These results imply that counterfactual statistical decision theory constitutes a coherent decision-making framework, which enables decision-makers to encode a variety of ethical and other subjective considerations.
Our axiomatic results also reconcile apparent inconsistencies and paradoxes that have been discussed in the literature.

Our analysis complements the broader literature on causal decision theory, which studies preference representations over interventions induced by structural causal models and utilities \citep[e.g.,][]{halpern2024representation}. In contrast, we characterize preferences over policy--state pairs induced by counterfactual utilities within the given potential outcome framework. Our framework can nest common models in behavioral economics, such as regret theory \citep{bell1982regret, loomes1982regret, bikhchandani2011transitive, lanzani2022correlation}, while preserving the von Neumann--Morgenstern axioms on the potential outcome space.

Finally, some scholars have expressed skepticism about counterfactual reasoning in decision-making and recommend that practitioners use standard utilities instead. \citet{dawid2023personalised} categorically reject the counterfactual approach, which they label as ``dangerously misguided'' and argue that it ``should not be
used in practice''. \citet{sarvet2023perspectives} raise identification and ontological objections, arguing that 
``when a counterfactual framework is deployed to determine social policies and regulations, it coerces conformity to an unverifiable metaphysics'' since ``a serious problem is that we have no direct evidence that these principal strata
exist''.

We argue that the role of statisticians is not to prescribe which
utility functions decision-makers ought to adopt, or on which outcome
space preferences should be defined. These choices embody the
subjective preferences, ethical judgments, or legal requirements of
decision-makers. Rather, the statistician’s task is to identify optimal decisions and quantify the associated uncertainty under a given utility specification.

\newpage
\bibliography{ref}

\newpage

\appendix

\setcounter{equation}{0}
\setcounter{figure}{0}
\setcounter{table}{0}
\setcounter{section}{0}
\setcounter{theorem}{0}
\setcounter{axiom}{0}
\setcounter{lemma}{0}
\renewcommand {\theequation} {S\arabic{equation}}
\renewcommand {\thefigure} {S\arabic{figure}}
\renewcommand {\thetable} {S\arabic{table}}
\renewcommand {\thesection} {S\arabic{section}}
\renewcommand {\thetheorem} {S\arabic{theorem}}
\renewcommand {\theaxiom} {S\arabic{axiom}}
\renewcommand {\thelemma} {S\arabic{lemma}}

\providecommand*{\theHtheorem}{}
\renewcommand*{\theHtheorem}{supp.theorem.\arabic{theorem}}

\providecommand*{\theHaxiom}{}
\renewcommand*{\theHaxiom}{supp.axiom.\arabic{axiom}}

\providecommand*{\theHlemma}{}
\renewcommand*{\theHlemma}{supp.lemma.\arabic{lemma}}

\begin{center}
\LARGE {\bf Supplementary Appendix}
\end{center}

\section{Within-state Comparisons}\label{app:fixed-P}

In this appendix section, we study the restricted within-state preference relation in Equation~\eqref{eqn:pref_within_state}. We show that this preference relation has weaker identification properties, as it does not identify the state. Here, identification asks what can be recovered by observing all comparisons in the preference relation under the maintained axioms. Even if the state is known, the within-state comparison only identifies the utility up to additive functions of potential outcomes and covariates.

Fix $P \in \Delta(\cY^\cD \times \cX)$. Recall the preference relation in Equation~\eqref{eqn:pref_within_state},
\begin{align*}
    \pi \succsim_{(\tilde u,P)} \rho \iff P^\pi \succsim_{\tilde u} P^\rho,
\end{align*}
for policies $\pi, \rho: \cY^\cD \times \cX \to \Delta(\cD)$. As explored in \cite{karni1983state}, \cite{rubin1987weak} and \cite{AumannSavage71}, this model treats both $\tilde u$ and $P$ as parameters of the decision maker where $P$ is the subjective belief about the population. It generalizes the Subjective Expected Utility model of \cite{savage1972foundations} and \cite{anscombe1963definition}. 

To study identification in this model, note that Definition~\ref{def:repr} and Equation~\eqref{eqn:value_fct} imply
\begin{align*}
     \pi \succsim_{(\tilde u,P)} \rho \iff V_P(\pi; \tilde u) \geq V_P(\rho; \tilde u) \iff V_Q(\pi; \tilde v) \geq V_Q(\rho; \tilde v),
\end{align*}
for any utility-state pair $(\tilde v, Q)$ such that $\tilde u(d; \by, \bx) \cdot P(\by, \bx) = \tilde v(d; \by, \bx) \cdot Q(\by, \bx)$. In particular, for any $Q \in \Delta(\cY^\cD \times \cX)$ such that $P \ll Q$, setting
\[
    \tilde v(d; \by, \bx) := \tilde u(d; \by, \bx) \frac{P(\by, \bx)}{Q(\by, \bx)},
\]
generates the same choices between policies. Hence, the subjective probability, as well as the utility, are unidentified in this model.

For completeness, we discuss axioms for the preference relation $\succsim_{(\tilde u,P)}$, henceforth denoted by $\succsim^*$ to indicate the non-identification of $P$. Because this relation is a restriction of the across-state preference relation in Definition~\ref{def:repr}, it inherits its axioms. These axioms are still strong enough to yield a representation.

In the following, for two policies $\pi, \rho: \cY^\cD \times \cX \to \Delta(\cD)$, the mixed decision rule $\alpha \pi + (1-\alpha)\rho$ assigns probability $\alpha \pi(d; \by, \bx) + (1-\alpha) \rho(d; \by, \bx)$ to decision $d$ given $(\by, \bx)$.

\begin{axiom}[Completeness]\label{axm:fixed_p_completeness}
For all policies $\pi,\rho$, either $\pi\succsim^*\rho$ or
$\rho\succsim^*\pi$.
\end{axiom}

\begin{axiom}[Transitivity]\label{axm:fixed_p_transitivity}
For all policies $\pi,\rho,\sigma$, if $\pi\succsim^*\rho$ and
$\rho\succsim^*\sigma$, then $\pi\succsim^*\sigma$.
\end{axiom}

\begin{axiom}[Independence]\label{axm:fixed_p_independence}
For all policies $\pi,\rho,\sigma$ and $\alpha\in(0,1]$, if
$\pi\succ^*\rho$, then $\alpha\pi+(1-\alpha)\sigma
    \succ^*
    \alpha\rho+(1-\alpha)\sigma$.
\end{axiom}

\begin{axiom}[Continuity]\label{axm:fixed_p_continuity}
For all policies $\pi,\rho,\sigma$, if
$\pi\succ^*\rho\succ^*\sigma$, then there exist
$\alpha,\beta\in(0,1)$ such that $\alpha\pi+(1-\alpha)\sigma
    \succ^*\rho
    \succ^*
    \beta\pi+(1-\beta)\sigma$.
\end{axiom}


\begin{theorem}[Proposition~7.4 of \cite{kreps1988notes}]\label{thm:fixed_p_vnm}
\label{thm:fixed_p_vnm-2}
The following are equivalent:
\begin{enumerate}
    \item $\succsim^*$ satisfies
    Axioms~\ref{axm:fixed_p_completeness}--\ref{axm:fixed_p_continuity};
    \item there exists a counterfactual utility $\tilde u:\cZ\to\R$ and a probability $P\in \Delta(\cY^\cD\times \cX)$, such that 
    for all policies $\pi, \rho: \cY^\cD \times \cX \to \Delta(\cD)$,
    \[
        \pi\succsim^*\rho
        \iff
        V_P(\pi;\tilde u)\geq V_P(\rho;\tilde u).
    \]
\end{enumerate}
Moreover, suppose $(\tilde u, P)$ represent $\succsim^*$.  Then, $(\tilde v, Q)$ represents $\succsim^*$ if and only if there exists a constant $a>0$ and a mapping $b:\cY^\cD\times \cX \to \bR$ such that 
\[
\tilde v(d, \by, \bx) Q(\by, \bx) = a\tilde u(d, \by, \bx) P(\by, \bx) + b(\by, \bx)
\]
for all $(d,\by,\bx)\in\cZ$. In particular, if $P$ has full support, then $\tilde v$ represents
$\succsim^*$ with the same $P$ if and only if there exist a constant
$a>0$ and a function $c:\cY^\cD\times\cX\to\R$ such that
\[
    \tilde v(d;\by,\bx)
    =
    a\tilde u(d;\by,\bx)+c(\by,\bx)
\]
for all $(d,\by,\bx)\in\cZ$.
\end{theorem}

Theorem~\ref{thm:fixed_p_vnm} formalizes the non-identification property of utility-state pair under the model in Equation~\eqref{eqn:pref_within_state}. In particular, utilities are identified only up to additive functions of principal strata and covariates. This is problematic for decision-making with counterfactual utility, which relies on interpreting utilities across units with different potential outcomes, such as comparing $\tilde u(d;\by,\bx)$ with $\tilde u(d;\by^\prime,\bx^\prime)$.

\section{An Alternative Formulation with Unit Strata as States}\label{app:unit_state}

In statistical decision theory, a state of nature is the unknown object that determines utilities and indexes the sampling distribution of the data \citep{wald1950statistical,berger1985statistical}. In treatment-choice problems, it is natural to take the population law as the state of nature, as is done explicitly in the literature \citep[e.g.,][]{manski2004statistical,STOYE200970,stoye2011axioms,koch2025statisticaldecisiontheorycounterfactual}.

However, one might instead take a unit's potential-outcome stratum as the primitive state and assess decisions by their performance for that unit's stratum. This formulation is mathematically equivalent to the fixed-$P$ comparison in Equation~\eqref{eqn:pref_within_state} and Appendix~\ref{app:fixed-P}, where $P$ acts as a prior over unit strata with the decision-maker adopting a Bayesian decision-criterion.

We briefly recall notation and terminology from \citet{stoye2011axioms, stoye2011statistical} building on \citet{anscombe1963definition}. Let $\cS$ denote the set of states of nature. As in Section~\ref{sec:setup}, take $\cZ = \cD \times \cY^\cD \times \cX$ to be the consequence space, evaluated by a counterfactual utility. An \emph{act} $f$ is a map from states to outcome distributions $f:\cS \to \Delta(\cZ)$. Given a counterfactual utility $\tilde u: \cZ \to \R$, an act $f$ is evaluated by its statewise expected utility
\[
    \tilde u \circ f(s) := \int_{\cZ} \tilde u(z) df(s)(z).
\]
To connect this notation to the setup in Section~\ref{sec:setup} take $\cS = \Delta(\cY^\cD \times \cX)$. Associate each act $f$ with a policy $\pi: \cY^\cD \times \cX \to \Delta(\cD)$ through $f_\pi(P) = \pi \cdot P$ for $P \in \cS$. Then
\[
    \tilde u \circ f_\pi(P) = \sum_{(d, \by, \bx) \in \cZ} \tilde u(d; \by, \bx) \pi(d; \by, \bx) P(\by, \bx) = V_P(\pi; \tilde u).
\]
Thus, in this formulation, acts and policies are equivalent objects to be evaluated with respect to the underlying population law.

Now, instead, take $\cS = \cY^\cD \times \cX$ to be the set of unit-level potential-outcome strata and covariates. Under Lemma~\ref{lem:truthful_acts}, for a state $s = (\by, \bx)$, every act is of the form
\begin{align}\label{eqn:stratum_act_representation}
    f_\pi(\by, \bx) = \sum_{d \in \cD} \pi(d; \by, \bx) \delta_{(d, \by, \bx)}, 
\end{align}
for some policy $\pi: \cY^\cD \times \cX \to \Delta(\cD)$. For any $A\subseteq\cZ$, $ f_\pi(\by, \bx)(A) = \sum_{d \in \cD} \pi(d; \by, \bx) \indicator{(d, \by, \bx) \in A}$. Given a counterfactual utility $\tilde u$, the statewise expected utility is
\begin{align*}
    \tilde u \circ f_\pi(\by, \bx) =  \sum_{d \in \cD} \tilde u(d; \by, \bx) \pi(d; \by, \bx).
\end{align*}
Here, acts are again equivalent to policies, but these are now evaluated based on the unit's potential outcome vector and covariates. However, the statewise expected utility still depends on the unit's stratum $(\by, \bx)$, which is not identified, even with standard utilities.

Thus, to make decision-making feasible and to recover Equation~\eqref{eqn:value_fct} (the decision-making criterion used in the literature), one must introduce a measure $P \in \Delta(\cY^\cD \times \cX)$ and integrate over states
\[
    \int_{\cY^\cD \times \cX} (\tilde u \circ f_\pi(\by, \bx)) dP = V_P(\pi; \tilde u).
\]
Formally, $P$ plays the role of a prior over states \citep{stoye2011axioms}. Thus, under the stratum formulation, evaluating acts (equivalently, policies) by $V_P(\pi;\tilde u)$ corresponds to Bayesian decision-making. Here $P$ is fixed, so comparisons are made within a fixed probability measure over unit strata.

Furthermore, in this setting, all acts in Equation~\eqref{eqn:stratum_act_representation} are \emph{truthful} acts, i.e., the decision is evaluated with respect to the state in which it was made. Because $|\cY| \geq 2$, there are no \emph{constant} truthful acts, i.e., acts always depend on the state $s = (\by, \bx)$. Consequently, axioms imposed on constant acts like \emph{Independence of irrelevant alternatives (IIA) for constant acts} in \citet{stoye2011axioms} or \emph{Independence for Constant Acts (ICA)} and \emph{Independence of Irrelevant Alternatives for Constant Acts} in \citet{stoye2011statistical} impose no restrictions. Moreover, we cannot pin down a vNM utility representation over constant acts like in \citet[Lemma~1]{stoye2011axioms}, which forms the basis for the analysis in \citet{stoye2011axioms, stoye2011statistical}.

Finally, this formulation confounds prior $P$ and utility $\tilde u$ and suffers from the same identification challenges as within-state comparisons discussed in Appendix~\ref{app:fixed-P}. Preferences over policies alone cannot generally separate the probability assigned to a stratum from its utility weight. Thus, even if $P$ is known, by Theorem~\ref{thm:fixed_p_vnm}, we can only identify a utility up to additive functions of principal strata and covariates,
\[
    \tilde u^\prime(d; \by, \bx) = a u(d; \by, \bx) + b(\by, \bx), \qquad a > 0.
\]
This is detrimental to the idea of decision-making with counterfactual utility, which relies on comparing utilities across units with different potential outcomes, such as comparing $\tilde u(d;\by,\bx)$ with $\tilde u(d;\by^\prime,\bx^\prime)$.

\section{Axiomatization for Utilities with Realized Outcomes}
\label{app:outcome_utility}

We consider the axiomatic results for settings, in which utilities depend on realized outcomes alone.
\begin{axiom}[Realized Outcome Sufficiency]\label{axm:outcome_sufficiency}
For all $p, q \in \Delta(Z)$, all $d,d^\prime \in D$ and $\bx,\bx^\prime \in \cX$, if
\begin{enumerate}
    \item $p(D = k) = \delta_d(k)$ and $q(D = k) = \delta_{d^\prime}(k), \quad \forall k \in \cD$,
    \item $p(\bX = \bm{v}) = \delta_{\bx}(\bm{v})$ and $q(\bX = \bm{v}) = \delta_{\bx^\prime}(\bm{v}), \quad \bm{v}\in \cX$,
    \item $p(Y(d)=y) = q(Y(d^\prime)=y), \quad y\in \cY$,
\end{enumerate}
then $p \sim q$.
\end{axiom}

Axiom~\ref{axm:outcome_sufficiency} implies indifference between $(\pi,P)$ and $(\rho,Q)$ whenever the distributions of the observed outcomes coincide, regardless of the decision or covariates. It is weaker than requiring indifference whenever two distributions agree on the marginal of $Y(D)$, but, as with Axiom~\ref{axm:irrelevance_of_untaken_decisions}, the two formulations are equivalent under Axioms~\ref{axm:completeness}--\ref{axm:continuity}. Moreover, Axiom~\ref{axm:outcome_sufficiency} implies Axiom~\ref{axm:irrelevance_of_untaken_decisions} and is thus also indifferent to unrealized outcomes.
The theorem below establishes that Axioms~\ref{axm:completeness}--\ref{axm:continuity},~and~\ref{axm:outcome_sufficiency} characterize the preferences induced on the potential outcome space $\cZ$ by utilities that depend on realized outcomes. The proof is given in Appendix~\ref{app:outcome}.
\begin{theorem}[Preferences of Outcome Utilities on the potential outcome Space]\label{thm:outcome}
The following are equivalent:
\begin{enumerate}
    \item $\succsim$ satisfies Axioms~\ref{axm:completeness}--\ref{axm:continuity} and~\ref{axm:outcome_sufficiency};
    \item there exists an outcome utility $u(d; \by, \bx) = u(y_d)$ that represents $\succsim$ in the sense of Definition~\ref{def:repr}.
\end{enumerate}
Moreover, the utility function $ u$ is unique up to positive affine transformation.
\end{theorem}

\section{Formal Results for Section~\ref{sec:failing_to_preserve_asymmetry}} \label{app:failing_to_preserve_asymmetry}

We formalize and generalize the discussion in Section~\ref{sec:failing_to_preserve_asymmetry} on \citet{Gelman_2025}'s proposal to extend counterfactual utilities by switching the order of evaluation. Lemma~\ref{lem:only_marginals} shows that, under this proposal, any two joint distributions with the same marginals induce the same decision, even if their dependence structure differs. Theorem~\ref{thm:gm_reduction_to_std} further shows that, under mild conditions on such extended utility, the resulting decision recommendation can always be obtained by a standard utility for every state of nature. These results are not restricted to binary outcomes.

\begin{lemma}\label{lem:only_marginals}
Suppose decisions are evaluated as described in Section~\ref{sec:setup_soe} following \citet{Gelman_2025}. If two states $P,Q \in \Delta(\cY^{\cD}\times\cX)$ satisfy
\[
    \E_P[Y(d)] = \E_Q[Y(d)] \qquad \forall d \in \cD,
\]
then, they induce the same decision.
\end{lemma}
Lemma~\ref{lem:only_marginals} follows directly from the fact that if $p_d=\E_P[Y(d)]=\E_Q[Y(d)] = q_d$ then $\tilde u^{\ext}(d; p_0,\ldots,p_{K-1}) = \tilde u^{\ext}(d; q_0,\ldots,q_{K-1})$. This shows that the joint law of $(Y(0), \ldots, Y(K-1))$ is irrelevant under GM's proposal.
Lemma~\ref{lem:only_marginals} is not the only implication for the joint structure of potential outcomes. Under GM's proposal and mild conditions on the extended utility $\tilde u^{\ext}$, one can in fact construct a standard utility $u^{\ext}(d;p_d)$ that induces the same decision behavior. We show this next.

We restrict attention to binary decisions but allow for general outcomes.
Fix a counterfactual utility $\tilde u$ and its extension $\tilde u^{\ext}$. Let $p_d=\E[Y(d)]$ for $d\in\{0,1\}$. Under GM's procedure summarized in Section~\ref{sec:setup_soe}, the only decision-relevant object is the contrast
\[
    \Gamma \tilde u^{\ext}(p_0,p_1)
    := \tilde u^{\ext}(1; p_0,p_1) - \tilde u^{\ext}(0; p_0,p_1),
\]
recommending $d=1$ if $\Gamma \tilde u^{\ext}(p_0,p_1)>0$, $d=0$ if $\Gamma \tilde u^{\ext}(p_0,p_1)<0$, and indifference if $\Gamma \tilde u^{\ext}(p_0,p_1)=0$. We impose the following assumptions.

\begin{assumption}[Bounded Means]\label{ass:bounded_means}
    There exist $L < U$ such that $p_d \in I := [L, U]$ for $d \in \{0, 1\}$.
\end{assumption}

\begin{assumption}[Decision Monotonicity]\label{ass:decision_monotonicity}
    For every fixed $p_0 \in I$, the map $p_1 \to \Gamma \tilde u^{\ext}(p_0, p_1)$ is non-decreasing on $I$. Similarly, for every fixed $p_1 \in I$, the map $p_0 \to \Gamma \tilde u^{\ext}(p_0, p_1)$ is non-increasing on $I$.
\end{assumption}

\begin{assumption}[Unique Crossing]\label{ass:unique_crossing}
Fix $p_0 \in I$.
\begin{enumerate}[label=(\alph*)]
    \item the map $p_1 \to \Gamma \tilde u^{\ext}(p_0, p_1)$ has at most one root.
    \item suppose that there exist $\underline{p}_1 < \overline{p}_1$ such that
        $\Gamma \tilde u^{\ext} (p_0, \underline{p}_1) < 0$ and $\Gamma \tilde u^{\ext}(p_0,\overline{p}_1) > 0$ hold.
    Then, there exists $\tau(p_0) \in (\underline{p}_1, \overline{p}_1)$ such that $\Gamma \tilde u^{\ext}(p_0, \tau(p_0)) = 0$.
\end{enumerate}
\end{assumption}
Assumption~\ref{ass:bounded_means} holds whenever outcomes are bounded, whereas Assumption~\ref{ass:decision_monotonicity} requires that a utility contrast is more likely to recommend a decision $d$ when $p_d$ is greater.
Lastly, Assumption~\ref{ass:unique_crossing} is a mild regularity condition. If the contrast can recommend each decision for some values of $(p_0,p_1)$, it must cross zero exactly once.

The following theorem shows that by switching the order of evaluation, for a given counterfactual utility under Assumptions~\ref{ass:bounded_means}--\ref{ass:unique_crossing},, one can always find a standard utility with the same optimal policy recommendation for any state of nature $P$. In particular, the problem reduces to comparing
(possibly non-linear transformation of) marginal means, which does not respect the original asymmetric structure.
\begin{theorem}[Equivalent standard utility]\label{thm:gm_reduction_to_std}
Assume $\cD=\{0,1\}$. Suppose that decisions are evaluated as proposed in \citet{Gelman_2025}. Let $\tilde u$ be a counterfactual utility and $\tilde u^{\ext}$ be a chosen extension. Under Assumptions~\ref{ass:bounded_means}--\ref{ass:unique_crossing}, there exists an extended standard utility of the form $u^{\ext}(d;p_d)=\phi_d(p_d)$ that induces the same optimal decision as $\tilde u^{\ext}$. That is, for all state of natures $(p_0,p_1)\in I^2$,
    \begin{align*}
        \Gamma \tilde u^\ext(p_0, p_1) > 0 \iff \Gamma u^\ext(p_0, p_1) > 0,
    \end{align*}
    and $\Gamma \tilde u^\ext(p_0, p_1) = 0 \iff \Gamma u^\ext(p_0, p_1) = 0$. Moreover, one may choose $\phi_1$ to be the identity, $\phi_1(p)=p$, and $\phi_0$ to be non-decreasing.
\end{theorem}
The proof is given in Appendix~\ref{app:gm_reduction_to_atd}.
We illustrate this result with two examples.
\begin{example}[Asymmetric Extension]\label{exp:asymmetric_extension}
Assume $Y\in\{0,1\}$. Then any (non-extended) counterfactual utility contrast can be written as
    \begin{align*}
    \Gamma u(y_0,y_1)
    = u_{10}\,\indicator{y_0>y_1}
    + u_{01}\,\indicator{y_0<y_1}
    + u_{00}\,\indicator{y_0=y_1=0}
    + u_{11}\,\indicator{y_0=y_1=1}.
\end{align*}
Under the asymmetric extension of Section~\ref{sec:non_unique_extension}, the extended contrast may take the form
\begin{align*}
    \Gamma \tilde u^\ext_{\mathrm{Asm}}(p_0,p_1)
    = \tilde u_{0}\,\indicator{p_0>p_1}
    + \tilde u_{1}\,\indicator{p_0<p_1}
    + \tilde u_{01}\,\indicator{p_0=p_1}.
\end{align*}
Note that this extension does not distinguish never-survivors, $(Y(0), Y(1)) = (0, 0)$, from always-survivors, $(Y(0), Y(1)) = (1, 1)$.
Assumption~\ref{ass:bounded_means} is immediate. Assumptions~\ref{ass:decision_monotonicity}~and~\ref{ass:unique_crossing} either require $\tilde u_0 \leq \tilde u_{01} \leq \tilde u_1 < 0$, $0 < \tilde u_0 \leq \tilde u_{01} \leq \tilde u_1$ or $\tilde u_0<0= \tilde u_{01} <\tilde u_1$. In particular, when $\tilde u_{01}=0$, the induced decision depends only on the sign of $p_1-p_0$ and can therefore be replicated by a standard utility of the form $u(d;y_d)=\beta y_d$.
\end{example}

\begin{example}[Product extension]\label{exp:product_extension}
In the setting of Example~\ref{exp:asymmetric_extension}, the product extension of Section~\ref{sec:non_unique_extension} yields the contrast
\begin{align*}
    \Gamma \tilde u^{\ext}_{\mathrm{Prod}}(p_0,p_1)
    = \tilde \lambda + \tilde \lambda_0 p_0 + \tilde \lambda_1 p_1 + \tilde \lambda_{01} p_0 p_1 .
\end{align*}
As discussed in Section~\ref{sec:non_unique_extension} the contrast is equivalent to assuming Assumption~\ref{ass:independence} in the context of Section~\ref{sec:setup}.
If $\tilde \lambda_{01}=0$, the induced decision rule can be replicated by a standard utility of the form $u(d;y_d)=\alpha_d +\beta_d y_d$. If $\tilde \lambda_{01}\neq 0$, no equivalent standard utility exists in general. However, under Assumptions~\ref{ass:decision_monotonicity}--\ref{ass:unique_crossing}, which imply $\max\{0, - \tilde \lambda_{01}\} \leq \tilde \lambda_1$ and $\tilde \lambda_0 \leq \min\{0, - \tilde \lambda_{01}\}$, one can still construct an equivalent extended standard utility under GM's procedure. For example, take $u^\ext(d;p_d)=\phi_d(p_d)$ with $\phi_1(p_1)=p_1$ and
\[
    \phi_0(p_0)
    =
    \begin{cases}
        -1,
        &
        \text{if }
        \tilde \lambda+\tilde \lambda_0 p_0
        +\tilde \lambda_1 p_1
        +\tilde \lambda_{01}p_0p_1>0
        \text{ for all }p_1\in [0, 1],
        \\[6pt]
        2,
        &
        \text{if }
        \tilde \lambda+\tilde \lambda_0 p_0
        +\tilde \lambda_1 p_1
        +\tilde \lambda_{01}p_0p_1<0
        \text{ for all }p_1\in [0, 1],
        \\[10pt]
        -\dfrac{\tilde \lambda+\tilde \lambda_0 p_0}
        {\tilde \lambda_1+\tilde \lambda_{01}p_0},
        &
        \text{otherwise.}
    \end{cases}
\]
Then, $\Gamma \tilde u^{\ext}_{\mathrm{Prod}}(p_0,p_1)>0$ if and only if $\phi_1(p_1)>\phi_0(p_0)$, so the two criteria induce the same recommendations.
\end{example}

\section{Mathematical Proofs}

\subsection{Proof of Theorem~\ref{thm:std}} \label{app:std}

We first show that the preference relation induced by a standard utility $u(d; y_d, \bx)$ satisfies Axioms~\ref{axm:completeness}--\ref{axm:irrelevance_of_untaken_decisions}. Since any standard utility is also a special case of counterfactual utilities on $\cZ$, Theorem~\ref{thm:vnm} implies that the induced preference relation satisfies Axioms~\ref{axm:completeness}--\ref{axm:continuity}. We will verify Axiom~\ref{axm:irrelevance_of_untaken_decisions}.  Let $P^d,Q^d \in \Delta(\cZ)$ satisfy the premise of Axiom~\ref{axm:irrelevance_of_untaken_decisions} with point mass at $D=d$ and $\bX=\bx$, and suppose $P(Y(d)=y)=Q(Y(d)=y)$ for all $y\in\cY$. Then, $P$ and $Q$ are degenerate in $\bX$, i.e., $P(\by,\bx') = P(\by)\,\delta_{\bx}(\bx')$ and $Q(\by,\bx') = Q(\by)\,\delta_{\bx}(\bx')$, and,
\begin{align*}
    V_P(d; u) &= \sum_{\by \in \cY^\cD} \sum_{\bx^\prime \in \cX} u(d; y_d, \bx^\prime) P(\by, \bx^\prime) \\
    &= \sum_{\by \in \cY^\cD} \sum_{\bx^\prime \in \cX} u(d; y_d, \bx^\prime) P(\by) \delta_{\bx}(\bx^\prime) \\
    &= \sum_{y_d \in \cY} u(d; y_d, \bx) P(y_d) \\
    &= \sum_{y_d \in \cY} u(d; y_d, \bx) Q(y_d) \\
    &= \sum_{\by \in \cY^\cD} \sum_{\bx^\prime \in \cX} u(d; y_d, \bx^\prime) Q(\by) \delta_{\bx}(\bx^\prime) \\
    &= \sum_{\by \in \cY^\cD} \sum_{\bx^\prime \in \cX} u(d; y_d, \bx^\prime) Q(\by, \bx^\prime) \\
    &= V_Q(d; u).
\end{align*}
Hence, we have $P^d \sim Q^d$, establishing Axiom~\ref{axm:irrelevance_of_untaken_decisions}.

Next, we show that any preference relation, which satisfies Axioms~\ref{axm:completeness}--\ref{axm:irrelevance_of_untaken_decisions}, admits a representation of a standard utility. By Theorem~\ref{thm:vnm}, Axioms~\ref{axm:completeness}--\ref{axm:continuity} imply the existence of a counterfactual utility $\tilde u(d;\by,\bx)$ such that
\begin{align}\label{eqn:std_proof_pref_rel}
    P^\pi \succsim Q^\rho \iff V_P(\pi;\tilde u) \ge V_Q(\rho;\tilde u).
\end{align}
We will show that under Axiom~\ref{axm:irrelevance_of_untaken_decisions}, $\tilde u$ must be a standard utility. Fix $d \in \cD$ and $\bx \in \cX$, and take $\by,\by^\prime \in \cY^{\cD}$ such that $y_d=y_d^\prime$. Consider the Dirac measures $\delta_{(d,\by,\bx)}$ and $\delta_{(d,\by',\bx)}$ in $\Delta(\cZ)$. By Axiom~\ref{axm:irrelevance_of_untaken_decisions},
\[
    \delta_{(d,\by,\bx)} \sim \delta_{(d,\by^\prime,\bx)}.
\]
Using the representation in Equation~\eqref{eqn:std_proof_pref_rel}, this implies
\[
    \tilde u(d;\by,\bx)
    = \E_{\delta_{(d,\by,\bx)}}[\tilde u(D;Y(0),\ldots,Y(K-1),\bX)]
    = \E_{\delta_{(d,\by^\prime,\bx)}}[\tilde u(D;Y(0),\ldots,Y(K-1),\bX)]
    = \tilde u(d;\by^\prime,\bx).
\]
Hence, for each fixed $(d,\bx)$, $\tilde u(d;\by,\bx)$ depends on $\by$ only through $y_d$. Define
\[
    \bar u(d;y_d,\bx) := \tilde u(d;\by,\bx),
\]
for any $\by \in \cY^{\cD}$ with $(\by)_d=y_d$  so that, for all $(d,\by,\bx)\in\cZ$, we have $\bar u(d;\by,\bx) =  \bar u(d;y_d,\bx)$.
By construction, $\bar u$ leaves Equation~\eqref{eqn:std_proof_pref_rel} unchanged, proving the existence of a standard utility representation.  Lastly, by Theorem~\ref{thm:vnm}, this standard utility representation is unique up to positive affine transformation. \qed

\subsection{Proof of Theorem~\ref{thm:add}}\label{app:add}

The proof proceeds similarly to that of Theorem~\ref{thm:std} presented above.
We first show that the preference relation induced by the utility of the form $\tilde u(d; \by, \bx) = \sum_{k \in \cD} u_k(d; y_k, \bx)$ satisfies Axioms~\ref{axm:completeness}--\ref{axm:continuity},~and~\ref{axm:indifference_to_correlation}.
Since any such utility is a special case of counterfactual utilities on $\cZ$, Theorem~\ref{thm:vnm} implies that the induced preference relation satisfies Axioms~\ref{axm:completeness}--\ref{axm:continuity}, leaving us only to verify Axiom~\ref{axm:indifference_to_correlation}.

Let us begin by assuming that $P^d,Q^{d} \in \Delta(\cZ)$ satisfy Axiom~\ref{axm:indifference_to_correlation} with point mass at $D = d$ and $\bX = \bx$, and $P(Y(k) = y) = Q(Y(k) = y)$ for all $k \in \cD, y \in \cY$. This implies $P(\by,\bx') = P(\by)\,\delta_{\bx}(\bx')$ and $Q(\by,\bx') = Q(\by)\,\delta_{\bx}(\bx')$. Thus,
\begin{align*}
    V_P(d; \tilde u) &= \sum_{\by \in \cY^\cD} \sum_{\bx^\prime \in \cX} \tilde u(d; \by, \bx^\prime) P(\by, \bx^\prime) \\
    &= \sum_{\by \in \cY^\cD} \sum_{\bx^\prime \in \cX} \tilde u(d; \by, \bx^\prime) P(\by) \delta_{\bx}(\bx^\prime) \\
    &= \sum_{\by \in \cY^\cD} \tilde u(d; \by, \bx) P(\by)  \\
    &= \sum_{\by \in \cY^\cD}\sum_{k \in \cD} u_k(d; y_k, \bx) P(\by)  \\
    &= \sum_{k \in \cD} \sum_{y_k \in \cY} u_k(d; y_k, \bx) P(Y(k) = y_k)  \\
    &= \sum_{k \in \cD} \sum_{y_k \in \cY} u_k(d; y_k, \bx) Q(Y(k) = y_k)  \\
    &= \sum_{\by \in \cY^\cD}\sum_{k \in \cD} u_k(d; y_k, \bx) Q(\by)  \\
    &= \sum_{\by \in \cY^\cD} \tilde u(d; \by, \bx) Q(\by)  \\
    &= \sum_{\by \in \cY^\cD} \sum_{\bx^\prime \in \cX} \tilde u(d; \by, \bx^\prime) Q(\by) \delta_{\bx}(\bx^\prime) \\
    &= \sum_{\by \in \cY^\cD} \sum_{\bx^\prime \in \cX} \tilde u(d; \by, \bx^\prime) Q(\by, \bx^\prime) \\
    &= V_Q(d; \tilde u).
\end{align*}
Hence, $P^d \sim Q^{d}$, satisfying Axiom~\ref{axm:indifference_to_correlation}.

We next show that any preference relation, which satisfies Axioms~\ref{axm:completeness}--\ref{axm:continuity},~and~\ref{axm:indifference_to_correlation}, admits a representation of the form $\tilde u(d; \by, \bx) = \sum_{k \in \cD} u_k(d; y_k, \bx)$. By Theorem~\ref{thm:vnm}, Axioms~\ref{axm:completeness}--\ref{axm:continuity} imply the existence of a counterfactual utility $\tilde u(d;\by,\bx)$ such that
\begin{align}\label{eqn:add_proof_pref_rel}
    P^\pi \succsim Q^\rho \iff V_P(\pi;\tilde u) \ge V_Q(\rho;\tilde u).
\end{align}
It remains to show that under Axiom~\ref{axm:indifference_to_correlation}, $\tilde u$ must be of the form $\sum_{k \in \cD} u_k(d; y_k, \bx)$.

Fix $(d,\bx)\in\cD\times\cX$ and define the restriction $\tilde u^{(d,\bx)}:\cY^{\cD}\to\R$ by $\tilde u^{(d, \bx)}(\by) = \tilde u(d; \by, \bx)$. Let $P,Q \in \Delta(\cY^{\cD})$ have the same marginals, and define $P^d = \delta_d \cdot P \cdot \delta_{\bx}$ and $Q^d = \delta_d \cdot Q \cdot \delta_{\bx}$ as elements of $\Delta(\cZ)$. By Axiom~\ref{axm:indifference_to_correlation}, $P^d \sim Q^d$. Hence, by Equation~\eqref{eqn:add_proof_pref_rel}, we have,
\begin{align*}
    \E_P[\tilde u^{(d, \bx)}(Y(0), \ldots, Y(K-1))] = \E_Q[\tilde u^{(d, \bx)}(Y(0), \ldots, Y(K-1))].
\end{align*}
Therefore, the conditions of Lemma~\ref{lem:Fishburn} below hold, implying the existence of functions
$\{u_k^{(d,\bx)}\}_{k\in\cD}$ such that $\tilde u^{(d, \bx)}(\by) = \sum_{k \in \cD} u_k^{(d, \bx)}(y_k)$. Define $u_k:\cD\times\cY\times\cX\to\R$ by $u_k(d;y,\bx):=u_k^{(d,\bx)}(y)$. Then
\begin{align*}
    \tilde u(d; \by, \bx) = \sum_{k \in \cD} u_k(d; y_k, \bx),
\end{align*}
which establishes the additive representation. The uniqueness of this representation up to positive affine transformation follows from the same argument used in the proof of Theorem~\ref{thm:std}. \qed

\subsection{Additivity Lemma}

The proof of Theorem~\ref{thm:add} relies on the following lemma which is implied by the results of \citet[Chapter~11]{fishburn1970utility}. For completeness, we provide a proof of this lemma, which is similar to the proof of Theorem~11.1 of \cite{fishburn1970utility}.
\begin{lemma}[Additive Utilities]\label{lem:Fishburn}
    Consider a utility function $\tilde u:\cY^{\cD}\to\R$. Suppose that for all $P,Q \in \Delta(\cY^{\cD})$ satisfying
    \[
        P(Y(k)=y)=Q(Y(k)=y), \qquad \forall k\in\cD,\ \forall y\in\cY,
    \]
    we have
    \[
        \E_P[\tilde u(Y(0),\ldots,Y(K-1))]
        =
        \E_Q[\tilde u(Y(0),\ldots,Y(K-1))].
    \]
    Then, there exist functions $u_k:\cY\to\R$  for $k \in \cD$ such that
    \[
        \tilde u(\by) = \sum_{k\in\cD} u_k(y_k),
        \qquad \forall \by\in\cY^{\cD}.
    \]
\end{lemma}

\begin{proof}
    Fix a baseline $\by^\prime \in \cY^\cD$. Assign arbitrary values to $u_0(y_0^\prime), \ldots, u_{K-1}(y_{K-1}^\prime)$ such that the following equality holds,
    \begin{align}\label{eqn:lem_fishburn_initial_value}
        \tilde u(\by^\prime) = \sum_{k \in \cD} u_k(y_k^\prime).
    \end{align}
    For each $k \in \cD$ and $a \in \cY$, define the function $u_k: \cY \to \R$ by
    \begin{align}\label{eqn:add_proof_uk_def}
        u_k(a) := \tilde u(y_0^\prime, \ldots, y_{k-1}^\prime, a, y_{k+1}^\prime, \ldots, y_{K-1}^\prime) - \sum_{l \neq k} u_{l}(y_{l}^\prime).
    \end{align}
   Now, fix any $\by \in \cY^\cD$. For each $k \in \cD$ define the probability measures
    \begin{align*}
        P_k &:= \frac{1}{2} \; \delta_{(y_0, \ldots, y_{k-1}, y_k^\prime, \ldots, y_{K-1}^\prime)} + \frac{1}{2} \; \delta_{(y_0^\prime, \ldots, y_{k-1}^\prime, y_{k}, y_{k+1}^\prime, \ldots, y_{K-1}^\prime)},
        \\
        Q_k &:= \frac{1}{2} \; \delta_{(y_0, \ldots, y_{k},y_{k+1}^\prime, \ldots, y_{K-1}^\prime)} + \frac{1}{2} \; \delta_{(\by^\prime)}.
    \end{align*}
    Since $P_k, Q_k \in \Delta(\cY^\cD)$ have the same marginals, it follows by assumption that $E_{P_k}[\tilde u(Y(0), \ldots, Y(K-1)] = E_{Q_k}[\tilde u(Y(0), \ldots, Y(K-1)]$. This implies the identity
    \begin{align*}
         &\tilde u(y_0, \ldots, y_{k-1}, y_k^\prime, \ldots, y_{K-1}^\prime) +  \tilde u(y_0^\prime, \ldots, y_{k-1}^\prime, y_{k}, y_{k+1}^\prime, \ldots, y_{K-1}^\prime)
         \\
         &= \tilde u(y_0, \ldots, y_{k},y_{k+1}^\prime, \ldots, y_{K-1}^\prime) + \tilde u(\by^\prime),
    \end{align*}
    or equivalently,
    \begin{align*}
        &\tilde u(y_0^\prime, \ldots, y_{k-1}^\prime, y_{k}, y_{k+1}^\prime, \ldots, y_{K-1}^\prime)
         \\
         &= \lt[\tilde u(y_0, \ldots, y_{k},y_{k+1}^\prime, \ldots, y_{K-1}^\prime) - \tilde u(y_0, \ldots, y_{k-1}, y_k^\prime, \ldots, y_{K-1}^\prime)\rt] + \tilde u(\by^\prime).
    \end{align*}
    Summing over $k = 0, \ldots, K-1$ telescopes,
    \begin{align*}
        \sum_{k = 0}^{K-1} \tilde u(y_0^\prime, \ldots, y_{k-1}^\prime, y_{k}, y_{k+1}^\prime, \ldots, y_{K-1}^\prime) = \tilde u(\by) + (K-1) \tilde u(\by^\prime).
    \end{align*}
    Inserting the left side into Equation~\eqref{eqn:add_proof_uk_def} and recalling Equation~\eqref{eqn:lem_fishburn_initial_value} yields,
    \begin{align*}
        \tilde u(\by) + (K-1) \tilde u(\by^\prime) &= \sum_{k \in \cD} \tilde u(y_0^\prime, \ldots, y_{k-1}^\prime, y_{k}, y_{k+1}^\prime, \ldots, y_{K-1}^\prime)
        \\
        &= \sum_{k \in \cD} u_k(y_k) + \sum_{k \in \cD} \sum_{d \neq k} u_{d}(y_{d}^\prime)
        \\
        &= \sum_{k \in \cD}  u_k(y_k) + (K-1) \sum_{k \in \cD} u_k(y_k^\prime)
        \\
        &= \sum_{k \in \cD}  u_k(y_k) + (K-1) \tilde u(\by^\prime).
    \end{align*}
    Canceling $(K-1) \tilde u(\by^\prime)$ on both sides,
    \begin{align*}
        \tilde u(\by) = \sum_{k \in \cD}  u_k(y_k).
    \end{align*}
    Since $\by \in \cY^\cD$ was arbitrary, this completes the proof.
\end{proof}

\subsection{Proof of Theorem~\ref{thm:outcome}}\label{app:outcome}

The proof proceeds similarly to that of Theorem~\ref{thm:std} presented above.
We first show that the preference relation induced by the utility of the form $u(y_d)$ satisfies Axioms~\ref{axm:completeness}--\ref{axm:continuity},~and~\ref{axm:outcome_sufficiency}.
Since any such utility is a special case of counterfactual utilities on $\cZ$, Theorem~\ref{thm:vnm} implies that the induced preference relation satisfies Axioms~\ref{axm:completeness}--\ref{axm:continuity}, leaving us only to verify Axiom~\ref{axm:outcome_sufficiency}.

Let us begin by assuming that $P^d,Q^{d^\prime} \in \Delta(\cZ)$ satisfies Axiom~\ref{axm:outcome_sufficiency}. That is, for $d, d^\prime \in \cD$, $\bx, \bx^\prime \in \cX$ we have
$P(\by,\bv) = P(\by)\,\delta_{\bx}(\bv)$ and $Q(\by,\bv) = Q(\by)\,\delta_{\bx^\prime}(\bv)$ and $P(Y(d) = y) = Q(Y(d^\prime) = y), y \in \cY$. Thus,
\begin{align*}
    V_P(d; u) &= \sum_{\by \in \cY^\cD} \sum_{\bv \in \cX} u(y_d) P(\by, \bv)
    \\
    &= \sum_{\by \in \cY^\cD} \sum_{\bv \in \cX} u(y_d) P(\by) \delta_{\bx}(\bv)
    \\
    &= \sum_{y \in \cY} u(y) P(Y(d) = y)
    \\
    &= \sum_{y \in \cY} u(y) Q(Y(d^\prime) = y)
    \\
    &= \sum_{\by \in \cY^\cD} \sum_{\bv \in \cX} u(y_{d^\prime}) Q(\by) \delta_{\bx^\prime}(\bv)
    \\
    &= \sum_{\by \in \cY^\cD} \sum_{\bv \in \cX} u(y_{d^\prime}) Q(\by, \bv)
    \\
    &= V_Q(d^\prime; u).
\end{align*}
Hence, $P^d \sim Q^{d^\prime}$, satisfying Axiom~\ref{axm:outcome_sufficiency}.

We next show that any preference relation, which satisfies Axioms~\ref{axm:completeness}--\ref{axm:continuity},~and~\ref{axm:outcome_sufficiency}, admits a representation of the form $u(y_d)$. By Theorem~\ref{thm:vnm}, Axioms~\ref{axm:completeness}--\ref{axm:continuity} imply the existence of a counterfactual utility $\tilde u(d;\by,\bx)$ such that
\begin{align}\label{eqn:outcome_proof_pref_rel}
    P^\pi \succsim Q^\rho \iff V_P(\pi;\tilde u) \ge V_Q(\rho;\tilde u).
\end{align}
It remains to show that under Axiom~\ref{axm:outcome_sufficiency}, $\tilde u$ must be of the form $u(y_d)$.

Take any $(d, \by, \bx), (d^\prime, \by^\prime, \bx^\prime) \in \cZ$ such that $y_d = y^\prime_{d^\prime}$. Consider the Dirac measures $\delta_{(d,\by,\bx)}$ and $\delta_{(d^\prime,\by^\prime,\bx^\prime)}$ in $\Delta(\cZ)$. By Axiom~\ref{axm:outcome_sufficiency},
$\delta_{(d,\by,\bx)} \sim \delta_{(d^\prime,\by^\prime,\bx^\prime)}$.
Using the representation in Equation~\eqref{eqn:outcome_proof_pref_rel}, this implies,
\begin{equation}\label{eqn:outcome_proof_equal}
\begin{aligned}
         \tilde u(d;\by,\bx)
        & = \E_{\delta_{(d,\by,\bx)}}[\tilde u(D;Y(0),\ldots,Y(K-1),\bX)] \\
        & = \E_{\delta_{(d^\prime,\by^\prime,\bx^\prime)}}[\tilde u(D;Y(0),\ldots,Y(K-1),\bX)]
        \\
        & = \tilde u(d^\prime,\by^\prime,\bx^\prime).
\end{aligned}
\end{equation}
Fix $y \in \cY$. Choose any $(d,\by,\bx)\in\cZ$ with $(\by)_d=y$ and define
\[
        \bar u(y) := \tilde u(d;\by,\bx),
\]
so that, for all $(d,\by,\bx)\in\cZ$, we have $\tilde u(d;\by,\bx)=\bar u(y_d)$. This is well defined because if $(d^\prime,\by^\prime,\bx^\prime)\in\cZ$ also satisfies $(\by^\prime)_{d^\prime}=y$, then, Equation~\eqref{eqn:outcome_proof_equal} implies $\tilde u(d;\by,\bx)=\tilde u(d^\prime;\by^\prime,\bx^\prime)$.
By construction, $\bar u$ leaves Equation~\eqref{eqn:outcome_proof_pref_rel} unchanged, proving the existence of  representation $u(y_d)$. The uniqueness of this utility up to positive affine transformation follows from the same arguments used in the proof of Theorem~\ref{thm:std}. \qed

\subsection{Truthful Acts Lemma}
\begin{lemma}[Truthful acts]\label{lem:truthful_acts}
Let $\cS$ denote the finite set of states of the world, let $\cD$ denote the finite set of decisions, and let $\cZ= \cD \times \cS$ be the outcome space. We call an act $f:\cS \to \Delta(\cZ)$ \emph{truthful} if
\[
    f(s)(\cD \times \{s\}) = 1.
\]
Then every truthful act is of the form
\[
    f_\pi(s)
    =
    \sum_{d\in\cD} \pi(d;s) \delta_{(d,s)},
\]
for a unique policy $\pi:\cS\to\Delta(\cD)$. 
\end{lemma}
Truthfulness encodes two conditions. First, a realized consequence $(d,s)\in\cZ$ records both the decision $d$ that was made and the state $s$ in which it was made. Second, decisions do not alter states. Both conditions hold when the unit-level stratum is taken as the state as the decision is made for that unit, and it cannot alter the unit's potential outcomes or covariates.

\begin{proof}
    Suppose $f$ is truthful. Define $\pi(d; s) = f(s)(\{(d, s)\}) \geq 0$ and observe
    \[
        \sum_{d \in \cD} \pi(d; s) = f(s)(\cD \times \{s\}) = 1.
    \]
    Consequently $\pi(\cdot; s) \in \Delta(\cD)$. Next, for $A \subseteq \cZ$,
    \begin{align*}
        \sum_{d\in\cD} \pi(d;s) \delta_{(d,s)}(A) = \sum_{d\in\cD} f(s)(\{(d, s)\})\delta_{(d,s)}(A) = f(s)(A),
    \end{align*}
    which proves $f(s) = \sum_{d\in\cD} \pi(d;s) \delta_{(d,s)}$. Uniqueness follows immediately by varying $A_z = \{z\}, z \in \cZ$.
\end{proof}

\subsection{Proof of Theorem~\ref{thm:gm_reduction_to_std}}
\label{app:gm_reduction_to_atd}

The proof follows from the following lemma whose proof is given below.
\begin{lemma}\label{lem:equivalence}
    Suppose that Assumptions~\ref{ass:bounded_means}--\ref{ass:unique_crossing} hold on $\Gamma \tilde u^\ext(p_0, p_1)$. Then, there exists non-decreasing functions $\phi_0$ and $\phi_1$ such that for all $(p_0, p_1) \in I^2$,
    \begin{align*}
         \Gamma \tilde u^\ext(p_0, p_1) > 0 \iff \phi_1(p_1) - \phi_0(p_0) > 0, \quad \Gamma \tilde u^\ext(p_0, p_1) = 0 \iff \phi_1(p_1) - \phi_0(p_0) = 0.
    \end{align*}
    Moreover, one can choose $\phi_1$ to be the identity, i.e., $\phi_1(p) = p$.
\end{lemma}
\begin{proof}
    For the ease of notation, we denote $\Gamma \tilde u^\ext(p_0, p_1)$ by $\Gamma \tilde u(p_0, p_1)$. Let $L < U$ be the bounds from Assumption~\ref{ass:bounded_means}.
    Set $\phi_1(p_1) = p_1$ and define the function $\phi_0: I \to [L-1, U+1]$ by
    \begin{align*}
        \phi_0(p_0) =
        \begin{cases}
             L-1 \ \ &\text{if } \Gamma \tilde u(p_0, p_1) > 0 \text{ for all } p_1\in I,
             \\
             U+1 \ \ &\text{if } \Gamma \tilde u(p_0, p_1) < 0 \text{ for all } p_1 \in I,
             \\
             \text{unique } \tau(p_0) \in I \text{ such that } \Gamma \tilde u(p_0, \tau(p_0)) = 0 \ \ & \text{otherwise}.
        \end{cases}
    \end{align*}
    Note that $\phi_0$ is well defined by Assumptions~\ref{ass:bounded_means}--\ref{ass:unique_crossing}. Since $\phi_1$ is non-decreasing, it remains to show that $\phi_0$ is non-decreasing and that $\mathrm{sign}(\Gamma \tilde u(p_0, p_1)) = \mathrm{sign}(\phi_1(p_1) - \phi_0(p_0))$.

    We first show that $\phi_0$ is non-decreasing. Take $p_0' < p_0''$, we need to show that $\phi_0(p_0') \leq \phi_0(p_0'')$. We consider the four cases.
\begin{enumerate}
    \item Suppose $\phi_0(p_0') = U+1$.
    Then, by definition $\Gamma \tilde u(p_0', p_1) < 0$ for all $p_1 \in I$.  By Assumption~\ref{ass:decision_monotonicity}, $p_0 \to \Gamma \tilde u(p_0, p_1)$ is non-increasing, implying $\Gamma \tilde u(p_0'', p_1) \leq \Gamma \tilde u(p_0', p_1) < 0$. Hence, $\phi_0(p_0'') = U + 1$ and thus $\phi_0(p_0') = \phi_0(p_0'')$.

    \item Suppose $\phi_0(p_0') = L - 1$.
    Since $\phi_0(p_0) \geq L-1$ for all $p_0 \in I$, it follows that $\phi_0(p_0') \leq \phi_0(p_0'')$.

    \item Suppose $\phi_0(p_0') \in [L, U], \phi_0(p_0'') \in \{L-1, U+1\}$.
    If $\phi_0(p_0'') = U+1$ then trivially $\phi_0(p_0') \leq \phi_0(p_0'')$. Now suppose $\phi_0(p_0'') = L-1$, then by definition $\Gamma \tilde u(p_0'', p_1) > 0$ for all $p_1 \in I$. By Assumption~\ref{ass:decision_monotonicity} and $p_0' < p_0''$, we must have $\Gamma \tilde u(p_0', p_1) \geq \Gamma \tilde u(p_0'', p_1) > 0$. Hence, $\Gamma \tilde u(p_0', p_1) > 0$ for all $p_1 \in I$ which implies $\phi_0(p_0') = L-1$, contradicting $\phi_0(p_0') \in [L, U]$.
    \item Suppose $\phi_0(p_0')$ and $\phi_0(p_0'') \in [L, U]$.
    Note that since $\phi_0(p_0'), \phi_0(p_0'') \in [L, U]$, we have $\Gamma \tilde u(p_0', \phi_0(p_0')) = 0$ and $\Gamma \tilde u(p_0'', \phi_0(p_0'')) = 0$. Suppose for contradiction that $\phi_0(p_0') > \phi_0(p_0'')$. Assumption~\ref{ass:decision_monotonicity} further implies
    \[
        0 = \Gamma \tilde u(p_0'', \phi_0(p_0'')) \leq \Gamma \tilde u(p_0'', \phi_0(p_0')) \leq \Gamma \tilde u(p_0', \phi_0(p_0')) = 0
    \]
    such that $\Gamma \tilde u(p_0'', \phi_0(p_0')) = 0$. By Assumption~\ref{ass:unique_crossing}(a) this imposes $\phi_0(p_0') = \phi_0(p_0'')$, yielding a contradiction.
\end{enumerate}
    This proves that $\phi_0$ is non-decreasing.
    Finally, we show that $\Gamma \tilde u(p_0, p_1)$ and $\phi_1(p_1) - \phi_0(p_0)$ have the same sign. Fix $p_0 \in I$. We consider the following cases.
\begin{enumerate}
    \item Suppose $\phi_0(p_0) = L-1$.
    This directly implies $\phi_1(p_1) - \phi_0(p_0) = p_1 - (L-1) > 0$ for all $p_1 \in I$. Moreover, by construction of $\phi_0$, we must have $\Gamma \tilde u(p_0, p_1) > 0$ for all $p_1\in I$.

    \item Suppose $\phi_0(p_0) = U+1$.
    In this case $\phi_1(p_1) - \phi_0(p_0) = p_1 - (U+1) < 0$ for all $p_1 \in I$. As in the previous case, the construction of $\phi_0$ implies $\Gamma \tilde u(p_0, p_1) < 0$ for all $p_1\in I$.

    \item Suppose $\phi_0(p_0) \in [L, U]$.
    By construction, $\Gamma \tilde u(p_0, \phi_0(p_0)) = 0$. If $p_1 < \phi_0(p_0)$, then $\phi_1(p_1) - \phi_0(p_0) < 0$. By Assumption~\ref{ass:decision_monotonicity}, we must have $\Gamma \tilde u(p_0, p_1) \leq \Gamma \tilde u(p_0, \phi_0(p_0)) = 0$. If $\Gamma \tilde u(p_0, p_1) = 0$, this would give a second root, which violates Assumption~\ref{ass:unique_crossing}. Hence, $\Gamma \tilde u(p_0, p_1) < 0$. The same argument applies to the case $p_1 > \phi_0(p_0)$. If $p_1 = \phi_0(p_0)$, we trivially have $\phi_1(p_1) - \phi_0(p_0) = p_1 - \phi_0(p_0) = 0$ and $\Gamma \tilde u(p_0, p_1) = \Gamma \tilde u(p_0, \phi_0(p_0)) = 0$.
\end{enumerate}
    We have shown that regardless of the value of $\phi_0(p_0)$, the signs of $\Gamma \tilde u(p_0, p_1)$ and $\phi_1(p_1) - \phi_0(p_0)$ agree. This concludes the proof.
\end{proof}

\end{document}